\documentclass[aps,pra,reprint,superscriptaddress]{revtex4-1}

\usepackage[utf8]{inputenc}
\usepackage{amsmath}
\usepackage{amssymb}
\usepackage{fullpage}
\usepackage{xcolor,array,bm}
\usepackage{amsthm}
\usepackage{hyperref}
\usepackage{cleveref}
\usepackage{braket}
\usepackage{tikz-cd}
\usepackage{nicefrac,xfrac}
\usepackage{longtable}

\newtheorem{theorem}{Theorem}[]
\newtheorem{lemma}[theorem]{Lemma}

\newtheorem{definition}[theorem]{Definition}
\newtheorem{prop}[theorem]{Proposition}
\newtheorem*{remark}{Remark}

\begin{document}

\title{Quantum algorithms for group convolution, cross-correlation, and equivariant transformations}
\author{Grecia Castelazo}\altaffiliation{Equal contribution.}
\author{Quynh T. Nguyen}\altaffiliation{Equal contribution.}
\affiliation{Department of Electrical Engineering and Computer Science, Massachusetts Institute of Technology, Cambridge, MA, USA}
\author{Giacomo De Palma}
\affiliation{Scuola Normale Superiore, Pisa, Italy}
\affiliation{Department of Mathematics, University of Bologna, Bologna, Italy}
\author{Dirk Englund}
\affiliation{Department of Electrical Engineering and Computer Science, Massachusetts Institute of Technology, Cambridge, MA, USA}
\affiliation{Research Laboratory for Electronics, Massachusetts Institute of Technology, Cambridge, MA, USA}
\author{Seth Lloyd}
\affiliation{Research Laboratory for Electronics, Massachusetts Institute of Technology, Cambridge, MA, USA}
\affiliation{Department of Mechanical Engineering, Massachusetts Institute of Technology, Cambridge, MA, USA}
\author{Bobak T. Kiani}
\affiliation{Department of Electrical Engineering and Computer Science, Massachusetts Institute of Technology, Cambridge, MA, USA}
\affiliation{Research Laboratory for Electronics, Massachusetts Institute of Technology, Cambridge, MA, USA}

\begin{abstract}
Group convolutions and cross-correlations, which are equivariant to the actions of group elements, are commonly used to analyze or take advantage of symmetries inherent in a given problem setting.
Here, we provide efficient quantum algorithms for performing linear group convolutions and cross-correlations on data stored as quantum states. Runtimes for our algorithms are poly-logarithmic in the dimension of the group and the desired error of the operation. Motivated by the rich literature on quantum algorithms for solving algebraic problems, our theoretical framework opens a path for quantizing many algorithms in machine learning and numerical methods that employ group operations.

\end{abstract}

\maketitle

\section{Introduction}
Symmetry and invariance are properties of functions that have central importance in mathematics and physics. In machine learning, for example, many successful algorithms exploit inherent symmetries in a problem to guide or bias an algorithm towards special classes of functions which are suitable for that problem. Notably, convolutional neural networks (CNNs) exploit the translationally invariant structures within images (say if a cat moves sideways in an image, the features of that cat move with it) 
\cite{albawi2017understanding,goodfellow2016deep}. 
Convolutional neural networks seek to take advantage of translational invariance; however, symmetries arise through countless other group actions which each correspond to different invariance properties. More recent literature has focused on generalizing the results of convolutional neural networks to a broader class of group actions via the analysis of equivariance \cite{cohen2016group,kondor2018generalization}. Informally, a function is 
\emph{equivariant} if it transfers symmetries from the function's input space into its output space. Convolutions and cross-correlations provide a means to apply linear equivariant transformations. For group convolutional neural networks, previous work have shown that equivariant functions are precisely those that implement some form of group convolution or cross-correlation \cite{kondor2018generalization,ravanbakhsh2017equivariance,maron2019universality}.

In this study, we overview the group equivariant transformations of group convolution and cross-correlation and present quantum algorithms to perform these linear group operations on data stored as quantum states. Our algorithms can output quantum states storing the output of a group convolution or cross-correlation operation with runtimes that scale polynomially with the condition number of the linear operation and poly-logarithmically with the dimension of the group. Our primary aim is to contribute a theoretical framework for quantizing existing classical algorithms or designing new quantum algorithms in the group theoretic setting. A brief discussion of some concrete application of our work is also included.

Given functions $f$ and $g$ which map group elements $u \in G$ to complex or real numbers, a convolution over a group $G$ is defined as
\begin{equation}
    (f \circledast g) (u) = \sum_{v \in G} f(uv^{-1})g(v).
    \label{eq:conv_group}
\end{equation}

Similarly, cross-correlation is defined as
\begin{equation}
    (f \star g) (u) = \sum_{v \in G} f(vu^{-1})g(v).
    \label{eq:cross-corr}
\end{equation}

As an example, let us consider $G$ as the cyclic group $\mathbb{Z}/n \mathbb{Z}$. Then the group operation is isomorphic to integer addition modulus $n$. For example, if $u$ and $v$ are two elements of the groups corresponding to integers $u$ and $v$ as well (abusing notation), then we note that the operation $uv$ is equivalent to $u+v \mod n$ and $uv^{-1}$ is equivalent to $u-v \mod n$. Hence, for the cyclic group, the convolution corresponds to the typical case for one dimensional functions:
\begin{equation}
    (f \circledast g) (u) = \sum_{v \in G} f(uv^{-1})g(v) = \sum_{v = 0}^{n-1} f(u-v)g(v),
\end{equation}
where indices above are taken mod $n$. 

Note, that the equations above hold for finite groups, which are the focus of this study. For general groups, performing group convolution would require integration over Haar measures \cite{kondor2018generalization,haarconvs}. Though classical algorithms exist for approximately performing equivariant transformations over infinite dimensional groups \cite{finzi2020generalizing,kondor2018clebsch,perraudin2019deepsphere,cohen2018spherical}, we leave this more general case to future work. 

Quantum algorithms have been proposed to efficiently solve many algebraic or group-theoretic problems \cite{childs2010quantum}. Motivated by this prior success, we aim to address here the question of whether quantum computers can efficiently perform linear group operations efficiently. For properly chosen oracles and inputs, we provide two different methods to perform the linear group operations above. The first implements linear group operations in the ``real" regime by applying well-known quantum algorithms for performing linear combinations of unitary operators \cite{kothari2014efficient}. The second implements convolution theorems in the group Fourier regime by leveraging efficient quantum algorithms for group Fourier transforms \cite{moore2006generic}. 

Our paper is organized as follows. First, we list some related works (\autoref{sec:ii}) and overview salient group theoretic and representation theory concepts essential to understanding our algorithms (\autoref{sec:iii-a}). We also introduce the main specific linear algebraic methods used in our study such as converting group operations into matrices (\autoref{sec:iii-b}) and block encoding (\autoref{sec:iii-c}).
Then, we present our algorithms for performing linear group operations both in the real (\autoref{sec:iv}) and Fourier (\autoref{sec:v}) regimes by providing block encodings -- for the linear operations on a quantum computer. Furthermore, in (\autoref{sec:vi}) we provide an algorithm for applying inverse convolutions or cross-correlations, i.e., deconvolution both in the real and Fourier regimes.
Finally, we give an example application of our methods in solving an integral equation exhibiting a specific symmetry (\autoref{sec:vii}) and conclude with some discussion of future work (\autoref{sec:viii}).

\section{Related works}
\label{sec:ii}

Motivation for this work derives from the study of algebraic problems in quantum computing and implementations of equivariant transformations in deep learning especially in the context of group convolutional neural networks. We catalog some of these related works here. 

\paragraph{Quantum algorithms for group theoretic problems} Prior motivation for solving group theoretic problems in quantum computing stems from quantum algorithms aimed at solving the hidden subgroup problem \cite{childs2010quantum}. \cite{van2006quantum} proposed an algorithm for solving the hidden shift problem which employed group deconvolution on quantum states storing a superposition of queried function values. These ideas were expanded in \cite{roetteler2016quantum,rotteler2010quantum}. \cite{moore2006generic} provides an algorithm to perform generic group Fourier transforms on a quantum computer which forms the basis for many of the transformations performed in this work. In the broader context of quantum circuit analysis, group convolution has been used to analyze rates of convergence of ensembles of unitaries \cite{emerson2005convergence,dankert2009exact}. 

\paragraph{Quantum algorithms for linear algebra} The specific methods we use in this study are based on algorithms for performing linear algebraic operations on a quantum computer. Methods for block encoding unitary operators \cite{low2019hamiltonian,gilyen2019quantum} and applying linear combinations of unitary matrices \cite{kothari2014efficient,gui2008duality} are extensively used in our algorithms. Prior work in quantum computing has proposed methods and algorithms for efficiently performing matrix multiplication or solving linear systems of equations for dense matrices. The most related papers are those for applying circulant or Toeplitz matrices \cite{wan2018asymptotic,zhou2017efficient,mahasinghe2016efficient}. Circulant matrices are a specific instance of the more general form of group cross-correlation matrices studied here. From a more applied perspective, other related work focuses on pre-conditioning matrices using circulant matrices or solving Green's functions by taking advantage of symmetries in a problem \cite{tong2020fast,shao2018quantum}.  We note that our work considers a different setting than that of \cite{lomont2003quantum} which proved that group convolution is ``physically impossible.'' The work of \cite{lomont2003quantum} assumed that the convolution filter is given as a quantum state, whereas here we assume oracle access to its entries.
 
\paragraph{Equivariant and group convolutional neural networks} In the past few years, many algorithms for equivariant neural networks have been proposed and analyzed \cite{kondor2018generalization,keriven2019universal,cohen2016group,ravanbakhsh2017equivariance}. These algorithms employ and analyze weight sharing schemes that are inherent in equivariant transformations. This work has motivated a long line of research aiming to take advantage of symmetries in data \cite{zaheer2017deep,dieleman2016exploiting,worrall2017harmonic} with applications particularly in physics and chemistry \cite{bogatskiy2020lorentz,thomas2018tensor,kondor2018clebsch,cohen2019gauge}.

Many quantum algorithms have converted machine learning algorithms into quantum algorithms that are related to convolutions. For example, \cite{kerenidis2019quantum} construct a quantum algorithm that mimics the operation of a classical convolutional neural network (e.g., for image recognition). Quantum versions of convolutional neural networks which parameterize convolutions as quantum gates have also been proposed \cite{cong2019quantum,liu2019hybrid,pesah2021absence}.

\section{Preliminaries}
\label{sec:iii}

\subsection{Background in representation theory and group Fourier transforms}
\label{sec:iii-a}

In this section, we will discuss how to perform a group Fourier transform via the irreducible representations of a group. For the discussion here, we restrict ourselves to finite groups where the exposition of representations and group Fourier transforms is simpler.




\paragraph{Representation} Representations of a group aim to translate the action of groups onto matrix operations. A representation of a group $G$ is a matrix valued function $\rho: G \to \mathbb{C}^{d_\rho \times d_\rho}$ such that $\rho(g_1)\rho(g_2) = \rho(g_1g_2)$ for all $g_1, g_2 \in G$ \cite{fulton2013representation}, where $d_\rho$ is the dimension of the representation. As an example, we have the trivial representation $\rho(g_i)=1$ for all $g_i \in G$. 

Representations of a group can arise when associating a basis vector to each element of a group. Let $V$ be a vector space of dimension $|G|$ with basis $\{e_x:x \in G\}$, then the left (right) action of any $g \in G$ is a permutation: $ge_x = e_{gx}$ ($ge_x = e_{xg}$). This gives rise to a representation consisting of permutation matrices of size $|G| \times |G|$ which are called the left and right regular representations, denoted $L_u$ and $R_u$ respectively for $u \in G$. These matrices permute the basis elements according to the left and right actions of the group: $L_i e_j = e_{ij}$ and $R_i e_j = e_{ji}$ \cite{fulton2013representation}. For example, the cyclic group $\mathbb{Z}/3\mathbb{Z}$ has the following $3 \times 3$ regular representation matrices (since this group is abelian, $R_i = L_i$ for all group elements):
\begin{equation}
    \begin{pmatrix}
    1 & 0 & 0 \\
    0 & 1 & 0 \\
    0 & 0 & 1
    \end{pmatrix}, \qquad
    \begin{pmatrix}
    0 & 0 & 1 \\
    1 & 0 & 0 \\
    0 & 1 & 0
    \end{pmatrix}, \qquad
    \begin{pmatrix}
    0 & 1 & 0 \\
    0 & 0 & 1 \\
    1 & 0 & 0
    \end{pmatrix}.
\end{equation}

A representation is unitary if $\rho(g)$ is a unitary matrix for all $g$. A representation is irreducible if it contains no proper invariant subspaces with respect to the action of the group. For finite groups, unitary irreducible representations always exist. On the contrary, a representation is reducible if it decomposes as a direct sum of irreducible subrepresentations. For example, if a representation $\rho$ can be decomposed into the direct sum of two other representations $\rho_1$ and $\rho_2$ as below,
\begin{equation}
    \rho(g) = Q^{-1} 
    \begin{pmatrix}
    \rho_1(g) & 0 \\ 0 & \rho_2(g)
    \end{pmatrix} Q,
\end{equation}
where $Q \in \mathbb{C}^{d_\rho \times d_\rho}$ is an invertible matrix, then it is reducible. Importantly, for compact groups, any representation $\rho$ can be decomposed as above into a direct sum of irreducible representations:
\begin{equation}
    \rho(g) = Q^{-1} \left[ \rho_1(g) \oplus \rho_2(g) \oplus \cdots \oplus \rho_k(g) \right] Q.
\end{equation}

For abelian groups, the irreducible representations all have dimension equal to one. For non-abelian groups, there is at least one irreducible representation which has dimension greater than one.

\paragraph{Group Fourier transform} Given any function $f: G \to \mathbb{C}$ which maps group elements to scalars, the group Fourier transform of $f$ is a function which maps irreducible representations to matrices whose output is denoted by $\hat{f}(\rho)$ and is defined as (for finite groups)
\begin{equation}
    \hat{f}(\rho) = \sum_{u \in G} f(u) \rho(u).
\label{eq:groupfour}
\end{equation}

As a corollary to the conventional Fourier transform, convolution in the Fourier regime of the group corresponds to matrix multiplication over irreducible representations. Namely, for convolutions we have
\begin{equation}
    \widehat{(f \circledast g)}(\rho) = \hat{f}(\rho) \hat{g}(\rho),
\end{equation}
Similarly for cross-correlations we have
\begin{equation}
    \widehat{(f \star g)}(\rho) = \hat{f}(\rho)^\dagger \hat{g}(\rho).
\end{equation}

Note that throughout this text, we will use $\circledast$ to indicate convolution and $\star$ to indicate cross-correlation.



\paragraph{Equivariance}
Equivariance is the property of a function that translates symmetries of a function from its domain to codomain (input to output domain). A function is equivariant when the action of the group commutes with the function. 

\begin{definition}[Paraphrased from \cite{kondor2018generalization}]
Let $G$ be a group and $\mathcal{X}_1, \mathcal{X}_2$ be two sets with corresponding G-actions
\begin{equation}
    T_g: \mathcal{X}_1 \to \mathcal{X}_1 \;\;\;\;\;\;\;\; T_g': \mathcal{X}_2 \to \mathcal{X}_2.
\end{equation}
Let $V_1$ and $V_2$ be vector spaces with basis elements labeled by elements of $\mathcal{X}_1$ and $\mathcal{X}_2$ respectively, and let $L_{V_1}$ ($L_{V_2}$) be the set of functions mapping $\mathcal{X}_1$ ($\mathcal{X}_2$) to $V_1$ ($V_2$). Let $\mathbb{T}$ and $\mathbb{T}'$ be the induced actions of group elements onto $V_1$ and $V_2$ respectively (\textit{i.e.,} permute vector elements). A map $\phi: L_{V_1} \to L_{V_2}$ is equivariant if
\begin{equation}
    \phi(\mathbb{T}_g(f)) = \mathbb{T}_g'(\phi(f)) \;\;\;\;\;\;\;\; \forall f \in L_{V_1}.
\end{equation}
\end{definition}

For example, circular or cyclic convolutional layers (\textit{i.e.,} filters or kernels wrap around to perform convolution) in a convolution neural network are equivariant to cyclic permutations -- applying a cyclic permutation to the pixels before or after the layer results in equivalent outputs \cite{ravanbakhsh2017equivariance}.

The property of equivariance can be visualized as a commutative diagram:
\begin{equation}
\begin{tikzcd}[row sep=large,column sep=huge]
L_{V_1}\arrow[r, "\mathbb{T}_g"]\arrow[d, "\phi"]& 
L_{V_1}\arrow[d, "\phi" ] \\
L_{V_2}\arrow[r, "\mathbb{T}_g'"]&  L_{V_2}
\end{tikzcd}
\end{equation}

Convolutions and cross-correlations are examples of functions which are equivariant to the actions of a group (see \autoref{app:equivariance} for more details). These operations are studied in deep learning under the topic of equivariant neural networks \cite{kondor2018generalization,ravanbakhsh2017equivariance}. \cite{kondor2018generalization} proves that a feedforward neural network layer is equivariant to the action of a group if and only if each layer of the neural network performs a generalized form of convolution or cross-correlation. 

\subsection{Converting group convolution to a linear algebraic formulation}
\label{sec:iii-b}

Since convolutions and cross-correlations are linear operations, one can convert them into a matrix formulation. One simple way to do so is to vectorize the two input functions into the transformation and represent the action of the transformation via a matrix. Given two functions that map group elements to complex numbers, $m,x: G \to \mathbb{C}$, we vectorize these functions over group elements by associating every group element to a basis of the vector space and denote the vectors of dimension $|G|$ as $\vec{m}$ (filter) and $\vec{x}$ (input). Any convolution or cross-correlation can be converted into a matrix weighted sum of the left or right regular representations $L_u$ and $R_u$ respectively.


\begin{widetext}

\begin{lemma}[Group operations as matrices]
Given a group $G$, let $\vec{m} \in \mathbb{C}^{|G|}$ and $\vec{x} \in \mathbb{C}^{|G|}$ be the filter and input for a group operation. Then, group convolutions and cross-correlations correspond to matrix weighted sums of the left or right regular representations. 
\begin{equation}
    \begin{split}
        (m \circledast x) (u) = \sum_{v \in G} m(uv^{-1})x(v) & \; \; \; \; \; \stackrel{{\rm convolution}}{\Longleftrightarrow} \; \; \; \; \; \vec{m} \circledast \vec{x} = M^{\circledast} \vec{x}, \; \; M^{\circledast} = \sum_{i \in G} m_i L_i \\
        (m \circledast_R x) (u) = \sum_{v \in G} m(v^{-1}u)x(v) & \; \; \; \; \; \stackrel{{\rm right \ convolution}}{\Longleftrightarrow} \; \; \; \; \; \vec{m} \circledast _R \vec{x} = M^{R \circledast} \vec{x}, \; \; M^{R \circledast} = \sum_{i \in G} m_i R_i \\
        (m \star x) (u) = \sum_{v \in G} m(vu^{-1})x(v) & \; \; \; \; \; \stackrel{{\rm cross-correlation}}{\Longleftrightarrow} \; \; \; \; \; \vec{m} \star \vec{x} = M^{\star} \vec{x}, \; \; \; M^{\star} = \sum_{i \in G} m_i L_i^{-1} \\
        (m \star_R x) (u) = \sum_{v \in G} m(u^{-1}v)x(v) & \; \; \; \; \; \stackrel{{\rm right \ cross-correlation}}{\Longleftrightarrow} \; \; \; \; \; \vec{m} \star _R \vec{x} = M^{R \star} \vec{x}, \; \; M^{R \star} = \sum_{i \in G} m_i R_i^{-1} \\
    \end{split}
    \label{eq:conv_to_linear_comb}
\end{equation}
\label{lem:conv_to_linear_comb}
\end{lemma}

For each of the operations above, we also have a corresponding convolution theorem which applies the operation in the Fourier domain of the group.

\begin{lemma}[Convolution theorems \cite{fulton2013representation}]
Given a group $G$, let $\vec{m} \in \mathbb{C}^{|G|}$ and $\vec{x} \in \mathbb{C}^{|G|}$ be the filter and input for a group operation. Let $\hat{m}(\rho)$ and $\hat{x}(\rho)$ indicate the value of the Fourier transform of the filter and input for irreducible representation $\rho$. Then, one can perform group operations in the Fourier regime by applying the corresponding convolution theorem.
\begin{equation}
    \begin{split}
        (m \circledast x) (u) = \sum_{v \in G} m(uv^{-1})x(v) & \; \; \; \; \; \stackrel{{\rm convolution}}{\Longleftrightarrow} \; \; \; \; \; \widehat{(m \circledast x)}(\rho) = \hat{m}(\rho) \hat{x}(\rho) \\
        (m \circledast_R x) (u) = \sum_{v \in G} m(v^{-1}u)x(v) & \; \; \; \; \; \stackrel{{\rm right \ convolution}}{\Longleftrightarrow} \; \; \; \; \; \widehat{(m \circledast_R x)}(\rho) =  \hat{x}(\rho) \hat{m}(\rho) \\
        (m \star x) (u) = \sum_{v \in G} m(vu^{-1})x(v) & \; \; \; \; \; \stackrel{{\rm cross-correlation}}{\Longleftrightarrow} \; \; \; \; \; \widehat{(m \star x)}(\rho) =  \hat{m}(\rho)^\dagger  \hat{x}(\rho)  \\
        (m \star_R x) (u) = \sum_{v \in G} m(u^{-1}v)x(v) & \; \; \; \; \; \stackrel{{\rm right \ cross-correlation}}{\Longleftrightarrow} \; \; \; \; \; \widehat{(m \star_R x)}(\rho) =  \hat{x}(\rho) \hat{m}(\rho)^\dagger  \\
    \end{split}
    \label{eq:convolution_theorems}
\end{equation}
\label{lem:convolution_theorems}
\end{lemma}

\end{widetext}

We defer proofs of the above to \autoref{app:conv_thm_proofs} and provide examples of these concepts in \autoref{app:rep_theory_ex}.

The above shows that there are two methods by which one can perform group convolution or cross-correlation on a quantum computer. The first is to apply the matrices described in \Cref{lem:conv_to_linear_comb} as a weighted sum of unitaries \cite{kothari2014efficient}. The second is to apply Fourier transforms to inputs and perform convolution in the Fourier regime as described in \Cref{lem:convolution_theorems}. This method takes advantage of the fact that many group Fourier transforms are efficiently performable on a quantum computer \cite{childs2010quantum,moore2006generic}. 

\subsection{Block encodings}
\label{sec:iii-c}

Throughout this study, we employ the block encoding framework to implement linear transformations on a quantum computer \cite{low2019hamiltonian}. In this framework, a desired linear but not necessarily unitary transformation $A \in \mathbb{C}^{2^w \times 2^w}$ bounded in the spectral norm by $\|A\|\leq 1$ is encoded in a unitary operator $U \in \mathbb{C}^{2^{(w+a)} \times 2^{(w+a)}}$ with $a$ ancilla qubits such that the top left block of $U$ is precisely $A$. 

\begin{equation}
U = \begin{pmatrix}
A & \cdot \\ \cdot & \cdot
\end{pmatrix}, \;\;\;\;\; \left(\bra{0^a} \otimes I_w \right) U \left( \ket{0^a} \otimes I_w \right) = A,
\end{equation}
where $I_w$ is the identity operation on the $w$ qubits encoding $A$. In other words, applying the unitary $U$ to a quantum state $\ket{0^a}\ket{\psi}$ and post-selecting on the measurement outcome $\ket{0^a}$ on the ancilla qubits is equivalent to applying the operation $A$ on $\ket{\psi}$.
\begin{equation}
    U \ket{0^a} \ket{\psi} = \ket{0^a} A \ket{\psi} + \ket{\text{garbage}},
\end{equation}
where $\ket{\text{garbage}}$ is a garbage state that is orthogonal to the subspace $\ket{0^a}$ (\textit{i.e.,} $[ \bra{0^a} \otimes I_w ] \ket{\text{garbage}} = 0$). The probability of successfully post-selecting $\ket{0}^a$ is thus equal to $\| A \ket{\psi} \|_2^2$.

\section{Quantum implementation as sum of unitaries}
\label{sec:iv}
Previous quantum algorithms have proposed efficient means to block encode and apply certain matrices as a sum of unitary matrices, each of which can be efficiently performed via quantum operations \cite{berry2015simulating,kothari2014efficient}. This framework can be applied to the form of the operations shown in \Cref{lem:conv_to_linear_comb} where convolution and cross-correlation operations are a weighted sum of the left or right regular (and unitary) representations. 

Let $w = \lceil \log_2 |G| \rceil$ indicate the number of qubits needed to block encode a given group operation. In the quantum case, we assume that we have access to either of the below oracles, $\mathcal{A}_m$ or $O_m$, as well as their inverses, which provide values of the convolution filter as $b$-bit descriptions or amplitudes of a quantum state:
\begin{equation}
\begin{split}
    \mathcal{A}_m &: \ket{0^w} \to \frac{1}{\sqrt{\|\vec{m}\|_1}} \sum_{i \in G} \sqrt{|m_i|} \ket{i}, \\
    O_m &: \ket{i} \ket{0^b} \to \ket{i} \ket{m_i}.
\end{split}
\end{equation}
If entries $m_i$ scale independently with the size of the group, one can efficiently convert oracle $O_m$ to $\mathcal{A}_m$ (and vice-versa up to phase factors) via algorithms for quantum digital-to-analog conversion \cite{mitarai2019quantum}, which we detail in \autoref{app:oracle_conversion} of the appendix. The oracle $O_m$ can be efficiently constructed if the entries $m_i$ are efficiently computable with a classical circuit, e.g. when $m$ is sparse or is the discretization of a kernel function (e.g., see \Cref{app:integraleqn}). We note that from the oracle $O_m$, one can also extract the phase of $m_i$, a fact which will become useful in block encoding the operations properly.


From here, we apply quantum algorithms for linear combinations of unitary matrices to perform group convolutions and cross-correlations. 
\begin{lemma}[Linear combination of unitaries, paraphrased from Lemma 2.1 of \cite{kothari2014efficient}]
    Let $V = \sum_i a_i U_i$ be a linear combination of unitary matrices $U_i$ with $a_i>0$. Let $A$ be a unitary matrix that maps $\ket{0^w}$ to $\frac{1}{\sqrt{a}} \sum_i \sqrt{a_i} \ket{i}$ where $a := \sum_i a_i$. Let $U := \sum_i \ket{i}\bra{i} \otimes U_i$, then $W := A^\dagger U A$ satisfies for any state $\ket{\psi}$
    \begin{equation}
        W \ket{0^w} \ket{\psi} = \sqrt{p} \ket{0^w} V \ket{\psi} + \ket{\Psi_\perp},
    \end{equation}
    where $p =  a^{-2}$ and the unnormalized state $\ket{\Psi_\perp}$ (depending on $\ket{\psi}$) satisfies $(\ket{0^w}\bra{0^w} \otimes I)\ket{\Psi_\perp} = 0$.
    \label{lem:linear_comb}
\end{lemma}
 In other words, \Cref{lem:linear_comb} shows that $W$ is a block encoding of the matrix $V$ \cite{gilyen2019quantum}. Performing group convolution or cross-correlation is a direct application of the above Lemma. 
\begin{lemma}[Block encoding of group convolution or cross-correlation]
\label{lem:O_m_block_encoding}
    Given oracle access to a filter $\vec{m}$ where $\vec{m}$ is normalized such that $\|\vec{m}\|_1 = 1$, one can block encode the matrix $M= \sum_i m_i U_i $ corresponding to group convolution or cross-correlation  (see \Cref{lem:conv_to_linear_comb} for the proper choice of $U_i$). This requires two calls to the oracle $\mathcal{A}_m$, one call to the oracle $O_m$, and efficient (classical) circuits for performing permutations based on group operations ($U_i$).
\end{lemma}
\begin{proof}
Based on which operation we would like to perform, we choose $U_i$ to be either $L_i$ or $R_i$ (or their inverses) as given in \Cref{lem:conv_to_linear_comb}. These permutation operations typically can be efficiently performed classically. Furthermore, we make a call to $O_m$ and apply a phase transformation to $U_i$ proportional to the phase of $m_i$. Finally, we use \Cref{lem:linear_comb}, setting $A$ in \Cref{lem:linear_comb} to $\mathcal{A}_m$ and $U_i$ to the chosen permutation operation (including the possible phase), and directly apply the results of \Cref{lem:linear_comb}.
\end{proof}

\begin{remark}
The normalization $\|\vec{m}\|_1 = 1$ is set to ensure that the largest singular value of the linear operation is no greater than 1. This bound can be easily obtained via the triangle inequality, \textit{e.g.,} for convolution $\|M^\circledast\| = \|  \sum_{i \in G} m_i L_i \| \leq \sum_{i \in G} |m_i| = \|\vec{m} \|_1$. This is required for block encoding a matrix within a larger unitary matrix. 
\end{remark}

The linear combination of unitaries approach can provide an efficient means to apply group operations to a quantum state as we describe below.

\begin{prop}[Applying group operations to an input state]
\label{prop:O_m_forward_convolution}
Given a quantum state $\ket{x} = \sum_i x_i \ket{i}$ containing the input state $\vec{x}$ normalized such that $\|\vec{x}\|_2 = 1$ and oracle access to the convolution filter $\vec{m}$, one can construct a state $\ket{m \circ x}$ which is equal to the normalized output of $\vec{m} \circ \vec{x}$ where $\circ$ corresponds to one of the group operations delineated in \Cref{lem:conv_to_linear_comb}. This operation has a runtime that scales as $O( T_B \|\vec{m} \circ \vec{x}\|_2^{-1})$ where $T_B$ is the runtime of the block encoding of \Cref{lem:linear_comb}.
\end{prop}
\begin{proof}
Let the matrix $M$ correspond to the linear operator where $\vec{m} \circ \vec{x} = M \vec{x}$ for a given group operation. Applying the block encoding of \Cref{lem:linear_comb}, we obtain the state
\begin{equation}
    \ket{0^w} M \ket{x} + \ket{x_\perp},
\end{equation}
where $\ket{x_\perp}$ is the ``garbage" projected into the perpendicular subspace. The operation is successful after measuring the first register and obtaining the outcome $\ket{0}$. The probability of success for this measurement is equal to $\|\vec{m} \circ \vec{x}\|_2^{2}$. By using amplitude amplification, this probability can be improved to $O(\|\vec{m} \circ \vec{x}\|_2)$ \cite{brassard2002quantum,ambainis2012variable}.
\end{proof}

\begin{remark} 
The runtime of the above operation is efficient when the term $\|\vec{m} \circ \vec{x}\|_2$, which be bounded by the condition number of the matrix corresponding to the group operation (see \Cref{prop:forward_convolution}), is small. These matrices are diagonalized (abelian groups) or block diagonalized (non-abelian groups) by the group Fourier transform as discussed in the next section. This provides a convenient method to calculate the condition number of any given linear operation.
\end{remark}

Note, that if the group $G$ is a cyclic group of order $n$, then the cross-correlation operation over the group produces a circulant matrix and we recover results similar to prior quantum algorithms for performing quantum matrix operations for circulant matrices \cite{zhou2017efficient,wan2018asymptotic}.

\section{Quantum implementation via convolution theorems}
\label{sec:v}
In this section, we consider quantum implementations of group operations when one is given oracle access to the convolutional filter $\vec{m}$ in the Fourier regime. Such settings are possible when for example one has access to the generating function of a kernel or can approximate the eigenvalues of the convolution matrix by analytically calculating a corresponding integral (see example provided later) \cite{wan2018asymptotic}. When performing group operations in the Fourier regime, there is an important distinction between abelian and non-abelian groups that arises from the property of the irreducible representations of each class of groups. Namely, abelian groups have the nice property that all of their irreducible representations are scalars. Furthermore, the Fourier transform for an abelian group can be easily obtained given the fact that any finite abelian group is a direct product of cyclic groups.

\begin{theorem}[Fundamental theorem of finite abelian groups \cite{dummit2004abstract}]
Every  finite abelian group is a direct product of cyclic groups whose orders are prime powers uniquely determined by the group.
\end{theorem}

Given this convenient theorem, the algorithm for performing abelian group operations is rather simple and we consider that case first. Then, we will generalize to the case of non-abelian groups which requires more detail.

\subsection{Block encoding for abelian groups}

Based on the fundamental theorem of finite abelian groups, one can form the Fourier transform for a finite abelian group by taking tensor products over the corresponding Fourier transform (DFT matrix) for the groups in the direct product. For example, if an abelian group $G$ is isomorphic to $k$ cyclic groups of dimension $d_i$ respectively, then
\begin{equation}
    F_G = \bigotimes_{i=1}^k F_{d_i} \; \; \; \; \; \; \;  \text{(abelian groups)},
\end{equation}
where $F_d$ is the discrete Fourier transform matrix of dimension $d$. This provides a direct means for diagonalizing convolutions and cross-correlations. For example, for convolution over an abelian group, we can form a matrix with the corresponding eigenvalues and eigenvectors.

\begin{equation}
\begin{split}
    F_G M^{\circledast} \vec{x} &= \sqrt{|G|} (F_G \vec{m}) \odot (F_G \vec{x})\\ &= \sqrt{|G|} \; \text{diag}(F_G \vec{m}) F_G \vec{x},
\end{split}
\end{equation}
where the $\odot$ is entry-wise multiplication. This implies that 

\begin{equation}
    M^{\circledast} = F_G^\dagger \; \text{diag}(\sqrt{|G|} F_G m) F_G,
    \label{eq:conv_mat_expanded}
\end{equation}
where the eigenvalues of $M^{\circledast}$ are the entries of $\sqrt{|G|} \; F_G m$ and the eigenvectors are the columns of $F_G$. Note, that in the above, we assume the $F_G$ are normalized to be unitary and hence we have the additional factor of $\sqrt{|G|}$ not typically seen in the convolution theorem. Since outputs are quantum states, this additional factor will be removed due to the normalization of the state.

Assume we are given access to an oracle $O_{\mathcal{F}m}$ which returns entries of $\hat{m}_i = \text{diag}(\sqrt{|G|} F_G m)_{ii}$ in a separate register:
\begin{equation}
    O_{\mathcal{F}m}: \ket{i}\ket{0^b} \to \ket{i}\ket{\hat{m}_i}.
    \label{eq:fourier_entry_oracle}
\end{equation}
This oracle can be efficiently constructed if the entries $m_i$ are efficiently computable with a classical circuit, e.g. when $m$ is sparse or when the group Fourier transform can be analytically computed (e.g., see \Cref{app:integraleqn}).

Any finite abelian group $G$ of size $n$ is isomorphic to a direct product of $c$ cyclic groups of dimension $n_1, \dots, n_c$. Therefore, the Fourier transform for a finite abelian group is simply $F_G = F_{n_1} \otimes \cdots \otimes F_{n_c}$ where $F_m$ is the standard unitary discrete Fourier transform matrix of dimension $m$. To apply a convolution matrix, we need to apply the Fourier transform, a diagonal matrix, and an inverse Fourier transform (see \autoref{eq:conv_mat_expanded}). The applications of the Fourier transform and inverse Fourier transform, at least for the abelian case, are applications of the corresponding quantum Fourier transform in the appropriate dimensions of the cyclic group decomposition. 

To perform the diagonal matrix operation $\text{diag}(\sqrt{|G|} F_G m)$, we can use a block encoding as below. Alternatively, we can directly apply theorems from \cite{tong2020fast} which provide a method for performing the inverse of a diagonal matrix. 

\begin{lemma}[Block encoding of diagonal matrix, adapted from Lemma 48 of \cite{gilyen2019quantum}]
\label{lem:block_encoding}
Let $A \in \mathbb{C}^{2^w \times 2^w}$ be a diagonal matrix and each entry of $A$ has absolute value of at most 1. Given access to the oracle $O_A$ such that
\begin{equation}
    O_A: \ket{i}\ket{0^b} \to \ket{i}\ket{A_{ii}},
\end{equation}
where $\ket{A_{ii}}$ is a $b$-bit binary description of diagonal element $i$, then we can implement a unitary block encoding $U$ such that $\|A - (\bra{0^{w+3}} \otimes I) U (\ket{0^{w+3}} \otimes I) \| \leq \epsilon$ with $O(\operatorname{polylog}\frac{1}{\epsilon}+w)$ gates and two calls to $O_A$. 
\end{lemma}



The block encoding for a given group operation is a simple application of the proper Fourier transforms and the lemma above. 
\begin{lemma}[Fourier block encoding of abelian group convolution or cross-correlation]
\label{lem:abelian_encoding}
For an abelian group $G$, let $w = \lceil \log_2 (|G|) \rceil$. Assume we are given oracle access $O_{\mathcal{F}m}$ to the convolution filter $\hat{m}$ in the Fourier regime as described earlier. Furthermore, assume the filter $\hat{m}$ is normalized so that $|\hat{m}(\rho)| \leq 1$ for all entries. Then, one can obtain a unitary operator $U$ that is a block encoding of the group operation, \textit{e.g.,} for convolution $\| M^{\circledast} - (\bra{0^{w+3}} \otimes I) U (\ket{0^{w+3}} \otimes I) \| \leq \epsilon$, with $O(\operatorname{polylog}\frac{1}{\epsilon}+w)$ additional gates and application of the group Fourier transform, the inverse group Fourier transform, and two calls to the oracle $O_{\mathcal{F}m}$.
\end{lemma}

\begin{proof}
Let us consider the case of group convolution. Other group operations are equivalent up to simple transformations in the elements of the diagonal transformation. 

For group convolution, one must perform the three operations given in \autoref{eq:conv_mat_expanded} copied below in unnormalized form:
\begin{equation}
     M^{\circledast} = F_G^\dagger \; \text{diag}(F_G \vec{m}) F_G.
\end{equation}

The operation $F_G$ and $F_G^\dagger$ are implementations of the proper quantum Fourier transform for the dimensions of the group. To perform diagonal matrix multiplication of matrix $A = \text{diag}(\sqrt{|G|} F_G \vec{m})$, we block encode $A$ into $U$ using oracle $O_{\mathcal{F}m}$ and \Cref{lem:block_encoding}. Thus, we apply the above operations in the order described to obtain the given block encoding. 
\end{proof}

\subsection{Block encoding for non-abelian groups}
For non-abelian groups, irreducible representations are matrices, and convolution applied in the Fourier regime requires matrix multiplication over the irreducible representation. In this setting, we now assume access to an oracle $O_{\mathcal{F}m}$ which provides matrix entries of the Fourier transform of a convolution filter in a given irreducible representation,

\begin{equation}
    O_{\mathcal{F}m}: \ket{\rho,a,b}\ket{0} \to \ket{\rho,a,b}\ket{\hat{m}(\rho)_{ab}},
\end{equation}
where $\rho \in \hat{G}$ indexes the irreducible representations and $\hat{m}(\rho)_{ab}$ is the $a,b$-th entry of the matrix $\hat{m}(\rho)$. Note, that this oracle can be simplified for abelian groups by simply removing the $a,b$ indexing since all irreducible representations are scalars.

Quantum algorithms efficiently perform group Fourier transforms over many non-abelian groups (\textit{e.g.,} dihedral and symmetric groups) \cite{moore2006generic,childs2010quantum}. The quantum group Fourier transform for a group $G$ returns a state containing a weighted superposition over irreducible representations \cite{childs2010quantum}:
\begin{equation}
\begin{split}
        F_G &= \sum_{x \in G} \ket{\hat{x}}\bra{x} \\&= \sum_{x \in G} \sum_{\rho \in \hat{G}} \sqrt{\frac{d_\rho}{|G|}} \sum_{j,k = 1}^{d_\rho} \rho(x)_{j,k} \ket{\rho,j,k}\bra{x},
\end{split}
\end{equation}
where $\ket{\hat{x}}$ is the group Fourier transform of a given basis vector $\ket{x}$ and $\hat{G}$ is the set of irreducible representations of $G$. The factor $\sqrt{d_\rho/|G|}$ included above enforces $F_G$ as unitary. $F_G$ also has the convenient property that it block diagonalizes the left and right regular representations into the irreducible representations \cite{childs2010quantum}, \textit{e.g.,} for the left regular representation, we have that \cite[Eq. 118]{childs2010quantum}.
\begin{equation}
    \hat{L}_i = \sum_{j \in G} \ket{\widehat{ij}} \bra{\hat{j}} = F_G L_i F_G^\dagger = \bigoplus_{\rho \in \hat{G}}\rho(i) \otimes I_{d_\rho} \,.
    \label{eq:block_diagonalizing_representation}
\end{equation}

Since convolutions and cross-correlations over non-abelian groups require matrix multiplication over irreducible representations, we cannot simply diagonalize the state written above as in the abelian case. Here, to convolve $\vec{m}$ with $\vec{x}$, we need to apply a matrix of the form below.
\begin{equation}
    \vec{m} \circledast \vec{x} = F_G^{-1} \left[\bigoplus_{\rho \in \hat{G}} \hat{m}(\rho) \otimes I_{d_\rho} \right] F_G \vec{x},
    \label{eq:nonabelian_general}
\end{equation}
where $\hat{m}(\rho)$ is the Fourier transformed matrix for irrep $\rho$ with dimensionality $d_\rho$.

\begin{lemma}[Fourier block encoding of general group convolution or cross-correlation]
\label{lem:nonabelian_encoding}
For a group $G$ of dimension $|G|$ which has irreducible representations of dimension no greater than $d_{\max}$, let $w = \lceil \log_2 (|G|) \rceil$. Assume we are given oracle access $O_{\mathcal{F}m}$ to the convolution filter $\hat{m}$ in the Fourier regime as described earlier. Furthermore, assume the filter $\hat{m}$ is normalized so that $|\hat{m}(\rho)_{ab}| \leq 1$ for all entries. Then, one can obtain a unitary operator $U$ that is a block encoding of the group operation -- \textit{e.g.,} for convolution $\| M^{\circledast} - d_{\max} (\bra{0^{w+3}} \otimes I) U (\ket{0^{w+3}} \otimes I) \| \leq \epsilon$ -- with one application of the group Fourier transform and its inverse, two calls to the oracle $O_{\mathcal{F}m}$, and $O(\operatorname{polylog}\frac{d_{\max}}{\epsilon}+w)$ additional gates.
\end{lemma}

\begin{proof}
Let us consider the case of group convolution. Other group operations are equivalent up to simple transformations in the elements of the block diagonal transformation. For group convolution, one must perform the three operations given in \autoref{eq:nonabelian_general} copied below:
\begin{equation}
      F_G^\dagger \left[\bigoplus_{\rho \in \hat{G}} \hat{m}(\rho) \otimes I_{d_\rho} \right] F_G .
\end{equation}

The operations $F_G$ and $F_G^\dagger$ are implementations of the group Fourier transform and its inverse \cite{moore2006generic}. Let $|G| \leq 2^w$ so we can encode the data in $w$ qubits. To form the block encoding we follow methods in \cite{gilyen2019quantum}. For our block encoding, we construct a data register of $w$ qubits and ancillary registers of $w+3$ qubits where the $\ket{0}$ measurement in this register corresponds to the location of the block encoding. We first apply the group Fourier transform to the data register. The middle operation $\left[\bigoplus_{\rho \in \hat{G}} \hat{m}(\rho) \otimes I_{d_\rho} \right]$ is a block diagonal matrix which we block encode using Lemma 48 of \cite{gilyen2019quantum}. Each row or column of the matrix is at most $d_{\max}$ sparse. Note, this lemma also requires oracles that provide the locations of each sparse entry in a given row or column of the matrix; in our case, since matrices are block diagonal, locating these entries is easy. Applying this operation up to error $\epsilon$ in operator norm requires two calls to the oracle $O_{\mathcal{F}m}$ and $O(\text{poly} \log \frac{d_{\max}}{\epsilon})$ additional gates \cite{gilyen2019quantum}. Finally, one applies an inverse group Fourier transform $F_G^\dagger$ to the data register to obtain the given encoding.
\end{proof}

\begin{remark}
$d_{\max}$ corresponds to the maximum sparsity of any row or column of the block diagonal matrix in our block encoding. The number of irreducible representations of a group is equal to the number of conjugacy classes of the group, so groups with many conjugacy classes tend to have lower dimensional irreducible representations.  For all abelian groups, $d_{\max}$ is trivially equal to $1$. For many non-abelian groups, $d_{\max}$ is also strictly bounded, \textit{e.g.,} $d_{\max} = 2$ for dihedral groups $D_{2n}$ for all $n$ \cite{childs2010quantum}.
\end{remark}

\subsection{Performing group operations on quantum states}

With the block encodings described above, we can apply linear group operations to an input state $\ket{x}$ and leverage the runtime benefits of the quantum group Fourier transform to efficiently perform linear group operations. Here, we consider the case of group convolution, and note that other group operations can be performed by simple changes to the steps below. First, we show how to perform group convolution directly on an input state.

\begin{prop}[Applying group convolution to $\ket{x}$]
\label{prop:forward_convolution}
For a group $G$ of dimension $|G|$ which has irreducible representations of dimension no greater than $d_{\max}$, assume we are given oracle access to the convolution filter $m$ to form block encodings described in \Cref{lem:abelian_encoding} or \Cref{lem:nonabelian_encoding}. Furthermore, assume the filter $\hat{m}$ is normalized so that $|\hat{m}(\rho)_{ab}| \leq 1$ for all entries. Given a quantum state $\ket{x} = \sum_i x_i \ket{i}$ containing the input state $\vec{x}$ normalized such that $\|\vec{x}\|_2 = 1$, one can construct a state $\ket{\tilde{y}}$ which is $\epsilon$-close to the true normalized output $\vec{m} \circ \vec{x}$, \textit{i.e.,} $\|\ket{\tilde{y}} - \ket{\vec{m} \circ \vec{x}} \| \leq \epsilon$ where $\circ$ corresponds to one of the group operations delineated in \Cref{lem:conv_to_linear_comb}. This operation has a runtime that scales as $O(  T_B \kappa d_{\max}/\|M\|)$ where $T_B$ is the runtime of the block encoding (which includes the dependence on $\epsilon$) of \Cref{lem:abelian_encoding} or \Cref{lem:nonabelian_encoding}, $\|M\|$ is the operator norm of $M$, and $\kappa$ is the condition number of $M$. 

\end{prop}
\begin{proof}
One first applies the block encoding of \Cref{lem:abelian_encoding} or \Cref{lem:nonabelian_encoding} to a state with the data encoded in the data register. After applying the block encoding, we obtain success when ancillary registers are measured in the $\ket{0}$ basis. The minimum singular value of the linear operation is $\|M\|/\kappa$ and the block encoding is such that it has a normalization factor of $d_{\max}$. Therefore, this has a worst-case success probability of $(\kappa \frac{d_{\max}}{\|M\|})^{-2}$. By using amplitude amplification, this probability can be improved to $O\left((\kappa \frac{d_{\max}}{\|M\|})^{-1}\right)$ \cite{brassard2002quantum,ambainis2012variable}. 
\end{proof} 

\begin{remark}
The condition number $\kappa$ can be calculated by analyzing the norms of the diagonal or block diagonal matrices in the block encoding. For example, for abelian groups,
\begin{equation}
    \kappa = \frac{\max_{\rho \in \hat{G}} |\hat{m}(\rho)|}{\min_{\rho \in \hat{G}} |\hat{m}(\rho)|}.
\end{equation}
For non-abelian groups, we analyze the singular values of the Fourier transform over its irreducible representations. Let $s_{min}(M)$ and $s_{max}(M)$ be the smallest and largest singular values of a matrix $M$, then for non-abelian groups,
\begin{equation}
    \kappa = \frac{\max_{\rho \in \hat{G}} s_{\max}(\hat{m}(\rho))}{\min_{\rho \in \hat{G}} s_{\min}(\hat{m}(\rho))}.
\end{equation}
\end{remark}

\section{Inverse group operations or deconvolution}
\label{sec:vi}
Given group operations as block encodings, we can conveniently perform polynomial transformations to the singular values of the block encoded matrix \cite{gilyen2019quantum}. This consequentially provides a straightforward means to apply inverse convolutions or cross-correlations, also commonly termed deconvolution. In this setting, one is provided with the output $\vec{y} = \vec{m} \circ \vec{x}$ as a quantum state $\ket{y}$ where $\circ$ corresponds to convolution or cross-correlation. Given information about the filter $\vec{m}$, one hopes to reconstruct the input to the group operation $\vec{x}$ as a quantum state $\ket{x}$. 

Here, we provide an algorithm to perform deconvolution given the block encoding using oracle $O_m$ described in \Cref{lem:O_m_block_encoding}. Similar steps can be followed to apply deconvolution for other block encodings.

\begin{prop}[Applying inverse group convolution or cross-correlation (deconvolution)]
For a group $G$ of dimension $|G|$, assume we are given oracle $O_m$ to access the convolution filter $m$ as described in \Cref{lem:O_m_block_encoding}. Given a quantum state $\ket{y}$ containing the output of $\vec{y} = \vec{m} \circ \vec{x}$ normalized such that $\|\vec{y}\|_2=1$, one can construct a state $\ket{\tilde{x}}$ which is $\epsilon$-close to the true normalized input $\vec{x}$, where $\circ$ corresponds to one of the group operations delineated in \Cref{lem:conv_to_linear_comb}. This operation has a runtime that scales as $O( T_B \frac{\kappa^2}{\|M\|} \operatorname{polylog} \frac{ \kappa}{\|M\| \epsilon} )$ where $T_B$ is the runtime of the block encoding (which includes the dependence on $\epsilon$) of \Cref{lem:O_m_block_encoding} and $\kappa$ is the condition number of the linear group operation $M$.
\label{prop:inverse_op}
\end{prop}

\begin{proof}
Follow the steps of \Cref{prop:O_m_forward_convolution}, but instead of directly applying the block encoding, apply a singular value transformation with polynomial approximations of the inverse function instead. In this setting, we start with a unitary $U$ that block encodes $M$ as in $\Cref{lem:O_m_block_encoding}$. Let $M=W D V^\dagger$ be the singular value decomposition of $M$ where $D$ is a diagonal matrix and $W$ and $V$ are unitary matrices. Using the quantum singular value transformation, with $d$ applications of the block encoding $U$ and its inverse, one can apply polynomial transformations $P(\cdot)$ of up to degree $d$ to the diagonal entries of $D$, \textit{i.e.,} $p(D)_{ii}=p(D_{ii})$ \cite{gilyen2019quantum}. Therefore, one must simply apply a polynomial approximating the inverse function as shown in Lemma 40 of \cite{gilyen2019quantum} and described below.

On the domain $[-1,1] \setminus \left(- \|M\| \kappa^{-1},\|M\| \kappa^{-1} \right)$, there exists an $O\left(\frac{\kappa}{\|M\|} \log(\frac{\kappa}{\|M\|  \epsilon})\right)$-degree polynomial function that is $\epsilon$-close to the inverse function $1/x$ \cite{gilyen2019quantum}. Using this polynomial transformation of the singular values, one can obtain an $\epsilon$-close block encoding of the deconvolution operation using $O\left(\frac{\kappa}{\|M\|} \log(\frac{\kappa}{\|M\|  \epsilon})\right)$ applications of the unitary $U$ block encoding $M$ and its inverse. See Lemma 40 and Theorem 41 of \cite{gilyen2019quantum}.  Upon application of this singular value transformation, the probability of successfully obtaining the state $\ket{\tilde{x}}$ is $O(\kappa^{-2})$ since the normalization factor of the inverse block encoding in Theorem 41 of \cite{gilyen2019quantum} is equal to the smallest eigenvalue of $M$, which is $\|M\| \kappa^{-1}$, and the smallest eigenvalue of the inverse operation (deconvolution) $M^{-1}$ is equal to $\|M\|^{-1}$. By using amplitude amplification, this probability can be improved to $O(\kappa^{-1})$ \cite{brassard2002quantum,ambainis2012variable}. 
\end{proof}

\begin{remark}
Alternatively, when performing deconvolution in the Fourier regime using the block encoding in \Cref{lem:nonabelian_encoding}, the operation has a runtime that scales as $O(\frac{T_B}{d_{\max}}\frac{\kappa^2}{\|M\|} \operatorname{polylog} \frac{ \kappa}{\|M\| \epsilon})$, where $T_B$ is the runtime of $\Cref{lem:nonabelian_encoding}$. We defer the proof and formal statement to \Cref{app:deconvolve}.
\end{remark}


\section{Example application for integral equations}
\label{sec:vii}
One particular application of the group convolution algorithms written above is in solving linear integral equations over domains that contain a certain symmetry. The example we consider here is solving a linear integral equation defined over the surface of an $n$-dimensional Torus $\boldsymbol{T}^d$ which is the $d$-fold product of the circle $\boldsymbol{S}^1$.

Consider an integral equation defined as below:
\begin{equation}
    g(\vec{t}) = f(\vec{t}) + \lambda \int_{\boldsymbol{T^d}} K(\| \vec{t} - \vec{t'} \|) f(\vec{t'}) d t' ,
\end{equation}
where $\vec{t} = [t_1, t_2, \dots, t_d]$ is an $d$-dimensional vector indicating the location of a point on the surface of the Torus, $f(\vec{t})$ is the unknown function (defined on the Torus) we would like to solve for, $\|\vec{t} - \vec{t'} \|$ is a norm which is cyclically invariant (see example in \autoref{app:integraleqn}), and $\lambda$ is a scalar constant. 

To solve the given integral equation numerically, one approximates the integral above via a numerical approximation scheme. For example, using the Nystr\"om method \cite{atkinson2009numerical}, one can approximate the integral as a weighted sum over $n^d$ discretized points with $n$ points evenly spaced along each dimension of $\vec{t}$. In each dimension $k$, we identify these evenly spaced points as $t_{k,i}$ where $i \in [n]$ such that any given discretized point $\vec{t}_{i_1, \dots, i_d} = [ t_{1,i_1}, \dots, t_{d,i_d}]$. Here, our approximation would be
\begin{equation}
\begin{split}
    \int_{\boldsymbol{T^d}} K(\| \vec{t} - \vec{t'} \|) & f(\vec{t'}) d t' \approx \\\frac{1}{n^d} \sum_{i_1 = 1}^{n} \cdots & \sum_{i_d=1}^{n} K(\| \vec{t} - \vec{t'}_{i_1, \dots, i_d} \| ) f(\vec{t'}_{i_1, \dots, i_d}).
\end{split}
\end{equation}
Therefore, to solve the above integral equation, we discretize our numerical solution and solve a linear systems of equations. In our example, we have
\begin{equation}
    \boldsymbol{g} = (I + \lambda \boldsymbol{K})\boldsymbol{f},
    \label{eq:ex_integral_equation}
\end{equation}
where bold text indicates discretized values vectorized in lexicographic order and $\boldsymbol{K}$ is a $n^d \times n^d$ matrix. The form of the matrix $\boldsymbol{K}$ corresponds to a group cross-correlation over the direct product of the cyclic groups $(\mathbb{Z}/n\mathbb{Z})^{(\times d)}$ where the filter $\vec{m}$ is equal to the values of the first row of $\boldsymbol{K}$. In other words, moving from one row of $\boldsymbol{K}$ to another is equivalent to applying a permutation operation corresponding to one of the group elements of $(\mathbb{Z}/n\mathbb{Z})^{(\times d)}$. 

Solving \autoref{eq:ex_integral_equation} requires inverting a dense matrix, and most existing quantum algorithms \cite{hhl, preconditionedHHL, childshhl} for this task only run efficiently with sparse matrices. This matrix, however, has added structure as it is a group cross-correlation and thus can be solved by inverting the cross-correlation operation via algorithms outlined here (\Cref{prop:inverse_op} or \Cref{app:deconvolve}). When $\boldsymbol{g}$ is provided as a quantum state, such an operation can be performed efficiently. In \autoref{app:integraleqn} we provide a detailed example of treating integrals over a Torus as group cross-correlation.

\section{Discussion}
\label{sec:viii}
Group convolutions and cross-correlations cover a wide class of equivariant functions commonly studied in machine learning and mathematics. We provide a framework and methodology for performing these equivariant group operations on a quantum computer. In well-conditioned cases, the runtimes of these operations scale logarithmically with the dimension of the group. Outputs of our algorithms, which are quantum states storing the vectorial output of the operations, can be post-processed or analyzed through various schemes, e.g., see \cite{lloyd2014quantum,kiani2022quantum,montanaro2015quantum,hoyer1997efficient} for examples.

For applications in machine learning, it has been shown that group-equivariant neural networks can be decomposed into layers of group convolutions \cite{kondor2018generalization} followed by nonlinear activation functions and pooling operators. Our algorithms provide a path towards quantizing the linear operations in group-equivariant neural networks \cite{cohen2016group,kondor2018generalization,ravanbakhsh2017equivariance,maron2019universality} and exploring potential quantum speedups in these machine learning models. Furthermore, one can apply our framework in a variational algorithm where a quantum circuit is parameterized and optimized as a convolutional filter \cite{cerezo2020variational}. More generally, our work provides a means to speed up linear operations for kernel matrices in the form of convolutions or cross-correlations commonly found in algorithms for machine learning and numerical methods \cite{rasmussen2003gaussian,dietrich1997fast,atkinson2009numerical}. This generalizes results from previous quantum algorithms for implementing circulant or Toeplitz matrices \cite{wan2018asymptotic,zhou2017efficient} and calculating Green's functions via convolutional formulations \cite{tong2020fast} using our algorithms for inverting group convolutions.

\section*{Acknowledgements}
This work was supported by the NSF, IARPA, DOE, and DARPA. The authors would like to thank Milad Marvian and William Kaminsky for helpful discussions.
GDP is a member of the ``Gruppo Nazionale per la Fisica Matematica (GNFM)'' of the ``Istituto Nazionale di Alta Matematica ``Francesco Severi'' (INdAM)''.


\appendix
\onecolumngrid
\section{Deferred proofs}
\label{app:conv_thm_proofs}

\subsection{Linear algebraic formulation
(\Cref{lem:conv_to_linear_comb})}
We would like to prove the equivalences shown in \Cref{lem:conv_to_linear_comb} and repeated below.
\begin{equation}
    \begin{split}
        (m \circledast x) (u) = \sum_{v \in G} m(uv^{-1})x(v) & \; \; \; \; \; \stackrel{\text{convolution}}{\Longleftrightarrow} \; \; \; \; \; \vec{m} \circledast \vec{x} = M^{\circledast} \vec{x}, \; \; M^{\circledast} = \sum_{i \in G} m_i L_i \\
        (m \circledast_R x) (u) = \sum_{v \in G} m(v^{-1}u)x(v) & \; \; \; \; \; \stackrel{\text{right convolution}}{\Longleftrightarrow} \; \; \; \; \; \vec{m} \circledast _R \vec{x} = M^{R \circledast} \vec{x}, \; \; M^{R \circledast} = \sum_{i \in G} m_i R_i \\
        (m \star x) (u) = \sum_{v \in G} m(vu^{-1})x(v) & \; \; \; \; \; \stackrel{\text{cross-correlation}}{\Longleftrightarrow} \; \; \; \; \; \vec{m} \star \vec{x} = M^{\star} \vec{x}, \; \; \; M^{\star} = \sum_{i \in G} m_i L_i^{-1} \\
        (m \star_R x) (u) = \sum_{v \in G} m(u^{-1}v)x(v) & \; \; \; \; \; \stackrel{\text{right cross-correlation}}{\Longleftrightarrow} \; \; \; \; \; \vec{m} \star _R \vec{x} = M^{R \star} \vec{x}, \; \; M^{R \star} = \sum_{i \in G} m_i R_i^{-1} \\
    \end{split}
\end{equation}

We provide a proof for convolution $(m \circledast x) (u)$ noting that the others require similar steps.

\begin{equation}
    \begin{split}
        (m \circledast x) (u) &= \sum_{v \in G} m(uv^{-1})x(v) \\
        &=  \sum_{v \in G} m(v^{-1})x(vu) \\
        &= \sum_{v \in G} m(v)x(v^{-1}u) \\
        &= \sum_{v \in G} m(v) \left[ L_v \vec{x} \right]_u,
    \end{split}
\end{equation}
where the notation $[ \cdot ]_i$ indicates the $i$-th component of the vector within the brackets. In the second and third lines above, we re-order the sum over all group elements by transforming $v \to vu$ and $v \to v^{-1}$ respectively. Converting the above into a vector form over the output, we have the final result:
\begin{equation}
    \vec{m} \circledast \vec{x} = \sum_{i \in G} m_i L_i \vec{x}.
\end{equation}

\subsection{Proofs of convolution theorems (\Cref{lem:convolution_theorems})}
We would like to prove the equivalences via the various convolution theorems shown in \Cref{lem:convolution_theorems} and repeated below.
\begin{equation}
    \begin{split}
        (m \circledast x) (u) = \sum_{v \in G} m(uv^{-1})x(v) & \; \; \; \; \; \stackrel{{\rm convolution}}{\Longleftrightarrow} \; \; \; \; \; \widehat{(m \circledast x)}(\rho) = \hat{m}(\rho) \hat{x}(\rho) \\
        (m \circledast_R x) (u) = \sum_{v \in G} m(v^{-1}u)x(v) & \; \; \; \; \; \stackrel{{\rm right \ convolution}}{\Longleftrightarrow} \; \; \; \; \; \widehat{(m \circledast_R x)}(\rho) =  \hat{x}(\rho) \hat{m}(\rho) \\
        (m \star x) (u) = \sum_{v \in G} m(vu^{-1})x(v) & \; \; \; \; \; \stackrel{{\rm cross-correlation}}{\Longleftrightarrow} \; \; \; \; \; \widehat{(m \star x)}(\rho) =  \hat{m}(\rho)^\dagger  \hat{x}(\rho)  \\
        (m \star_R x) (u) = \sum_{v \in G} m(u^{-1}v)x(v) & \; \; \; \; \; \stackrel{{\rm right \ cross-correlation}}{\Longleftrightarrow} \; \; \; \; \; \widehat{(m \star_R x)}(\rho) =  \hat{x}(\rho) \hat{m}(\rho)^\dagger  \\
    \end{split}
\end{equation}

For standard convolution, we have:
\begin{equation}
    \begin{split}
        \widehat{(m \circledast x)}(\rho) &= \sum_{u \in G} \rho(u) \sum_{v \in G} m(uv^{-1})x(v) \\
        &= \sum_{u \in G} \sum_{v \in G} \rho(u) \rho(v^{-1}) \rho(v) m(uv^{-1})x(v) \\
        &=  \sum_{v \in G} \sum_{u \in G} m(uv^{-1}) \rho(uv^{-1}) x(v) \rho(v) \\
        &= \sum_{v \in G} \left[ \sum_{u \in G} m(uv^{-1}) \rho(uv^{-1}) \right] x(v) \rho(v) \\
        &= \sum_{v \in G} \hat{m}(\rho)  x(v) \rho(v) \\
        &= \hat{m}(\rho) \hat{x}(\rho).
    \end{split}
\end{equation}

Since $(m \circledast_R x)(u) = (x \circledast m)(u) $, then the above argument can be applied to also show that $\widehat{(m \circledast_R x)}(\rho) =  \hat{x}(\rho) \hat{m}(\rho)$.

For standard cross-correlation, we similarly can show that:
\begin{equation}
    \begin{split}
        \widehat{(m \star x)}(\rho) &= \sum_{u \in G} \rho(u) \sum_{v \in G} m(vu^{-1})x(v) \\
        &= \sum_{u \in G} \sum_{v \in G} \rho(u)  \rho(v^{-1}) \rho(v)  m(vu^{-1})x(v) \\
        &=  \sum_{v \in G} \sum_{u \in G} m(vu^{-1}) \rho(uv^{-1}) x(v) \rho(v) \\
        &=  \sum_{v \in G} \left[ \sum_{u \in G} m(vu^{-1}) \rho(vu^{-1})^\dagger \right] x(v) \rho(v) \\
        &=  \sum_{v \in G} \hat{m}(\rho)^\dagger x(v) \rho(v) \\
        &=  \hat{m}(\rho)^\dagger \hat{x}(\rho).
    \end{split}
\end{equation}

For right cross-correlation, we have that:
\begin{equation}
    \begin{split}
        \widehat{(m \star_R x)}(\rho) &= \sum_{u \in G} \rho(u) \sum_{v \in G} m(u^{-1}v)x(v) \\
        &= \sum_{u \in G} \sum_{v \in G}  \rho(v) \rho(v^{-1}) \rho(u)    m(u^{-1}v)x(v) \\
        &=  \sum_{v \in G} \sum_{u \in G} x(v) \rho(v) m(u^{-1}v) \rho(v^{-1}u)  \\
        &=  \sum_{v \in G} x(v) \rho(v) \left[ \sum_{u \in G} m(u^{-1}v) \rho(u^{-1}v)^\dagger \right]  \\
        &=  \sum_{v \in G}  x(v) \rho(v) \hat{m}(\rho)^\dagger \\
        &=  \hat{x}(\rho) \hat{m}(\rho)^\dagger.
    \end{split}
\end{equation}

\subsection{Digital to analog oracle conversion}
\label{app:oracle_conversion}
As a reminder, we have analog ($\mathcal{A}_m$) and digital oracles ($O_m$) as shown below.
\begin{equation}
\begin{split}
    \mathcal{A}_m &: \ket{0^w} \to \frac{1}{\sqrt{\|\vec{m}\|_1}} \sum_{i \in G} \sqrt{|m_i|} \ket{i}, \\
    O_m &: \ket{i} \ket{0^b} \to \ket{i} \ket{m_i}.
\end{split}
\end{equation}

Given oracle $O_m$, our goal is to construct $\mathcal{A}_m$. Here, we assume that $O_m$ returns values normalized such that the magnitude of the maximum value of $m_i$ is equal to 1. This is chosen to maximize the success probability of oracle conversion which can be performed by following the steps below.

\begin{enumerate}
    \item Beginning with the state $\ket{0^w}\ket{0^b}$, obtain an equal superposition of states in the support of $m$. 
    \begin{equation}
        \ket{0^w}\ket{0^b} \to \frac{1}{\sqrt{|\operatorname{supp}(m)|}} \sum_{i \in \operatorname{supp}(m)} \ket{i}\ket{0^b},
    \end{equation}
    where $\operatorname{supp}(m)$ returns the set of basis states in the support of $m$. If the filter $m$ has full support, then this is equivalent to applying Hadamard gates to each qubit.
    \item Call oracle $O_m$ and perform (classical) transformations to obtain the magnitude of the filter resulting in
    \begin{equation}
        \frac{1}{\sqrt{|\operatorname{supp}(m)|}} \sum_{i \in \operatorname{supp}(m)} O_m \ket{i}\ket{0^b} \to \frac{1}{\sqrt{|\operatorname{supp}(m)|}} \sum_{i \in \operatorname{supp}(m)} \ket{i}\ket{|m_i|}.
    \end{equation}
    \item Append a qubit and conditionally rotate the qubit by $\sqrt{|m_i|}$.
    \begin{equation}
        \frac{1}{\sqrt{|\operatorname{supp}(m)|}} \sum_{i \in \operatorname{supp}(m)} \ket{i}\ket{m_i} \ket{0} \to \frac{1}{\sqrt{|\operatorname{supp}(m)|}} \sum_{i \in \operatorname{supp}(m)} \ket{i}\ket{m_i} \left( \sqrt{|m_i|} \ket{0} + \sqrt{1 - | m_i | } \ket{1} \right).
    \end{equation}
    \item Measuring the last appended register, the oracle conversion is successful when the outcome of the measurement is $\ket{0}$. We note, that this register need not be measured right away and can be included in the block encoding to be measured later. 
\end{enumerate}

The runtime of this procedure depends on the probability of successfully measuring the $\ket{0}$ state in the last step. This probability is $|\operatorname{supp}(m)|^{-1} \sum_{i \in \operatorname{supp}(m)} |m_i|$ and is equal to the average value of $|m_i|$. If values of $m_i$ are $\Theta(1)$ and do not decay with the dimension of the group, then this success probability is also $\Omega(1)$. Finally, additional gates are needed to obtain an equal superposition over states in the support of $m$ as in step 1. This, in most cases, requires a number of operations that scale poly-logarithmically with the dimension of the state. For example, for filters with support over all states, this is equivalent to applying Hadamard gates to each qubit. 

\subsection{Deconvolution in the Fourier regime (\Cref{prop:inverse_op})}\label{app:deconvolve}
To perform deconvolution in the Fourier regime, we apply Theorem 41 of \cite{gilyen2019quantum} to perform an inverse block encoding.
\begin{lemma} [Inverse block encoded matrix, adapted from Theorem 41 of \cite{gilyen2019quantum}] Let $A$ be an invertible matrix and $A = (\bra{0^{a}}\otimes I) U (\ket{0^{a}}\otimes I)$. Let $\delta$ be the smallest singular value of $A$ and $0 < \epsilon \leq \delta \leq \frac{1}{2}$. For $m = O(\frac{1}{\delta} \log {\frac{1}{\delta \epsilon}})$, there is an efficient circuit to implement  $U_{\Phi}$ such that
\begin{equation*}
    \left|(\bra{+}\otimes I) U_{\Phi} (\ket{+}\otimes I) - \frac{\delta}{2} A^{-1}\right| \leq \epsilon,
\end{equation*}
where $U_{\Phi}$ can be implemented using a single ancilla qubit and $O(m)$ gates, which include $m$ uses of $U$, $U^{\dagger}$, $C_{\ket{0^a}\bra{0^a}} NOT$ and single-qubit gates.
\end{lemma}

From here, one can perform deconvolution in the Fourier regime via block encodings described in \Cref{lem:abelian_encoding} and \Cref{lem:nonabelian_encoding}.

\begin{prop}[Applying deconvolution in the Fourier regime] For a group $G$, assume we are given oracle access $O_{\mathcal{F}m}$ to the Fourier transformed convolution filter $\hat{m}$ as described in \Cref{lem:nonabelian_encoding}. Given a quantum state $\ket{y}$ containing the output of $\vec{y} = \vec{m} \circ \vec{x}$ normalized such that $\|\vec{y}\|_2=1$, one can construct a state $\ket{\tilde{x}}$ which is $\epsilon$-close to the true normalized input $\vec{x}$, where $\circ$ corresponds to one of the group operations delineated in \Cref{lem:conv_to_linear_comb}. This operation has a runtime that scales as $O( T_B d_{\max} \frac{\kappa^2}{\|M\|} \operatorname{polylog} \frac{d_{\max} \kappa}{\|M\| \epsilon})$ where $T_B$ is the runtime of the block encoding (which includes the dependence on $\epsilon$ and $d_{\max}$) of \Cref{lem:nonabelian_encoding} and $\kappa$ is the condition number of the linear group operation $M$.
\label{app:fourier-deconvolve}
\end{prop}
\begin{proof}
    We take $A = M/d_{\max} =( \bra{0^{w+3}}\otimes I) U (\ket{0^{w+3}}\otimes I)$ where $U$ is the block encoding in \Cref{lem:nonabelian_encoding}. The smallest eigenvalue of $M/d_{\max}$ is $\delta = \|M\|\kappa^{-1} d_{\max}^{-1}$. Thus, one can obtain a block encoding of $A^{-1}$ in $O(T_B \frac{d_{\max}\kappa}{\|M\|}\log {\frac{d_{\max} \kappa}{\|M\|\epsilon}})$ operations. Upon application of this block encoded inverse, the probability of successfully obtaining the state $\ket{\tilde{x}}$ is $O(\kappa^2)$ since the normalization factor of the block encoding of the inverse $A^{-1}$ is $\delta/2 = O(\|M\|\kappa^{-1} d_{\max}^{-1})$ and the smallest singular value of $A^{-1}$ is $d_{\max} \|M\|^{-1}$. By using amplitude amplification \cite{brassard2002quantum,ambainis2012variable}, this probability of success can be improved to $O(\kappa)$.

    Note, that the assumption $\epsilon \leq \delta$, while being natural, can be removed by applying Corollary 69 of \cite{gilyen2019quantum}. Furthermore, the assumption $\delta \leq \frac{1}{2}$ can be fulfilled by rescaling entries of $\hat{m}$ accordingly.
\end{proof}

Note, that the runtime of the above algorithm has an additional factor of $d_{\max}$ compared to the runtime shown in \Cref{prop:inverse_op}.

\section{Representation theory of finite groups: example for the dihedral group}
\label{app:rep_theory_ex}
In this section we will provide a detailed example that highlights some important concepts of representation theory \cite{reptheory2004springer} focusing on the dihedral groups $D_n$, and specifically on $D_3$, which is the non-abelian group having the smallest group order.

\paragraph{Representations} 
A \emph{representation} of a finite group $G$ on a finite-dimensional complex vector space $V$ is a homomorphism $\rho: G \rightarrow GL(V)$ -where $GL(V)$ is the group of invertable linear transformations of the vector space $V$- of $G$ to the group of automorphisms of $V$ (invertible matrices).
\paragraph{Left and right regular representations}
If we associate each element $u$ of a group $G$ to a basis element $e_u$ in a vector space $V$, then the left and right regular representations, denoted by $L_u$ and $R_u$, respectively, are matrices that permute the basis elements according to the left and right actions of the group:
\begin{equation*}
    L_u e_v = e_{uv} \;\;\;\;\;\;\;\;\; R_u e_v = e_{vu}.
\end{equation*}

\paragraph{Subrepresentations} 
A \emph{subrepresentation} of a representation $V$ is a vector space $W$ of $V$ which is invariant under $G$.  
For compact groups, any representation $\rho$ can be decomposed as a direct sum of subrepresentations which are irreducible,
\begin{equation}
    \rho(g) = Q^{-1}[\rho_1(g)\oplus\rho_2(g)\oplus\cdots\oplus\rho_k(g)]Q.
\end{equation}
where $Q$ is an invertible matrix and each $\rho_i$ is irreducible.

\paragraph{Irreducible representations}
A representation is irreducible if there is no proper, nontrivial subspace of V that is invariant under the action of G.
It is called \emph{completely reducible} if it decomposes as a direct sum of irreducible subrepresentations.
The number of irreducible representations for a finite group is equal to the number of conjugacy classes \cite{fulton2013representation}. A \emph{conjugacy class} in $G$ is a nonempty subset $H$ of $G$ that is closed under the action of the group on itself by conjugation, {\it i.e.}

\begin{itemize}
    \item Given any $x,y \in H$, there exists $g \in G$ such that $gxg^{-1} = y$.
    \item If $x \in H$ and $g \in G$, then $gxg^{-1} \in H$.
\end{itemize}

Furthermore, the sum of the dimensions squared of all the irreducible representation of a group $G$ equals group size $G$.

\begin{equation}
    \sum_{\sigma} |d_{\sigma}|^2 = |G|
\end{equation}
 
\subsection{Representations of $D_n$}

\paragraph{Dihedral groups}
The dihedral group $D_n$ is the group of symmetries of the regular $n$-gon in the plane. The dihedral group $D_n$ is of order $2n$ and is represented by $D_n = \mathbb{Z}/n\mathbb{Z} \rtimes \mathbb{Z}/2\mathbb{Z}$, with the group law

\begin{equation*}
    (x,a)\cdot (y,b) = (x+(-1)^a y, a+b),
\end{equation*}
for $x,y \in \mathbb{Z}/n\mathbb{Z}$ and $a,b \in \mathbb{Z}/2\mathbb{Z}$. 

The dihedral group $D_{n}$ with $2n$ elements is isomorphic to a semidirect product of the cyclic groups $\mathbb{Z}/n\mathbb{Z}$ and $\mathbb{Z}/2\mathbb{Z}$. Let $r$ be the generator of $\mathbb{Z}/n\mathbb{Z}$ and $s$ be the generator of $\mathbb{Z}/2\mathbb{Z}$, then the dihedral group $D_n$ can be written compactly as
\begin{equation}
    \langle r, s | s^2=e, r^n = e, s r s^{-1} = r^{-1} \rangle.
\end{equation}

\paragraph{Irreducible representations of $D_n$}
For $n$ even, we have the following 1-dimensional irreducible representations
\begin{equation}
\begin{split}
    \sigma_{tt}((x,a)) & = 1 \\
    \sigma_{ts}((x,a)) & = (-1)^a \\
    \sigma_{st}((x,a)) & = (-1)^x \\
    \sigma_{ss}((x,a)) & = (-1)^{x+a}.
\end{split}
\label{eq:1d-irrep}
\end{equation}
For $n$ odd, we have $\sigma_{tt}$ and $\sigma_{ts}$ only. The $2$-dimensional irreducible representations are of the form

\begin{equation}
    \sigma_h((x,0)) = 
    \begin{pmatrix}
    e^{2\pi ihx/n} & 0 \\
    0 & e^{-2\pi ihx/n}
    \end{pmatrix}
    \hspace{0.8cm}
    \sigma_h((x,1)) = 
    \begin{pmatrix}
    0 & e^{2\pi ihx/n} \\
    e^{-2\pi ihx/n} & 0 
    \end{pmatrix},
    \label{eq:2d-irrep}
\end{equation}
for $h \in \{1,2, \ldots, \lceil\frac{n}{2}\rceil-1 \}$. The sum of the squared dimensions of the irreducible representations is equal to $2n$, which is the size of the group:
\begin{equation*}
    \sum_{\sigma} |d_{\sigma}|^2 = 2n = |G|.
\end{equation*}

\subsection{Representations of $D_3$}

The dihedral group $D_3$ is obtained by composing the six symmetries of an equilateral triangle. The dihedral group $D_3$ and the cyclic group $C_6$ are the only two groups that have order 6. Unlike $C_6$ (which is abelian),  $D_3$ is non-abelian. Products of group elements of $D_3$ are shown in the Cayley table shown in \autoref{table:cayley}.

Like all dihedral groups, group elements of $D_3$ are generated by $s$ and $r$, where $s$ is a rotation by  $\pi$ radians about an axis passing through the center and one of the vertices of a regular $n$-gon and $r$ is a rotation by $2\pi/n$ about the center of the $n$-gon (see \autoref{fig:d3}).

\begin{figure}[h]
\centering
\includegraphics[width=0.65\textwidth]{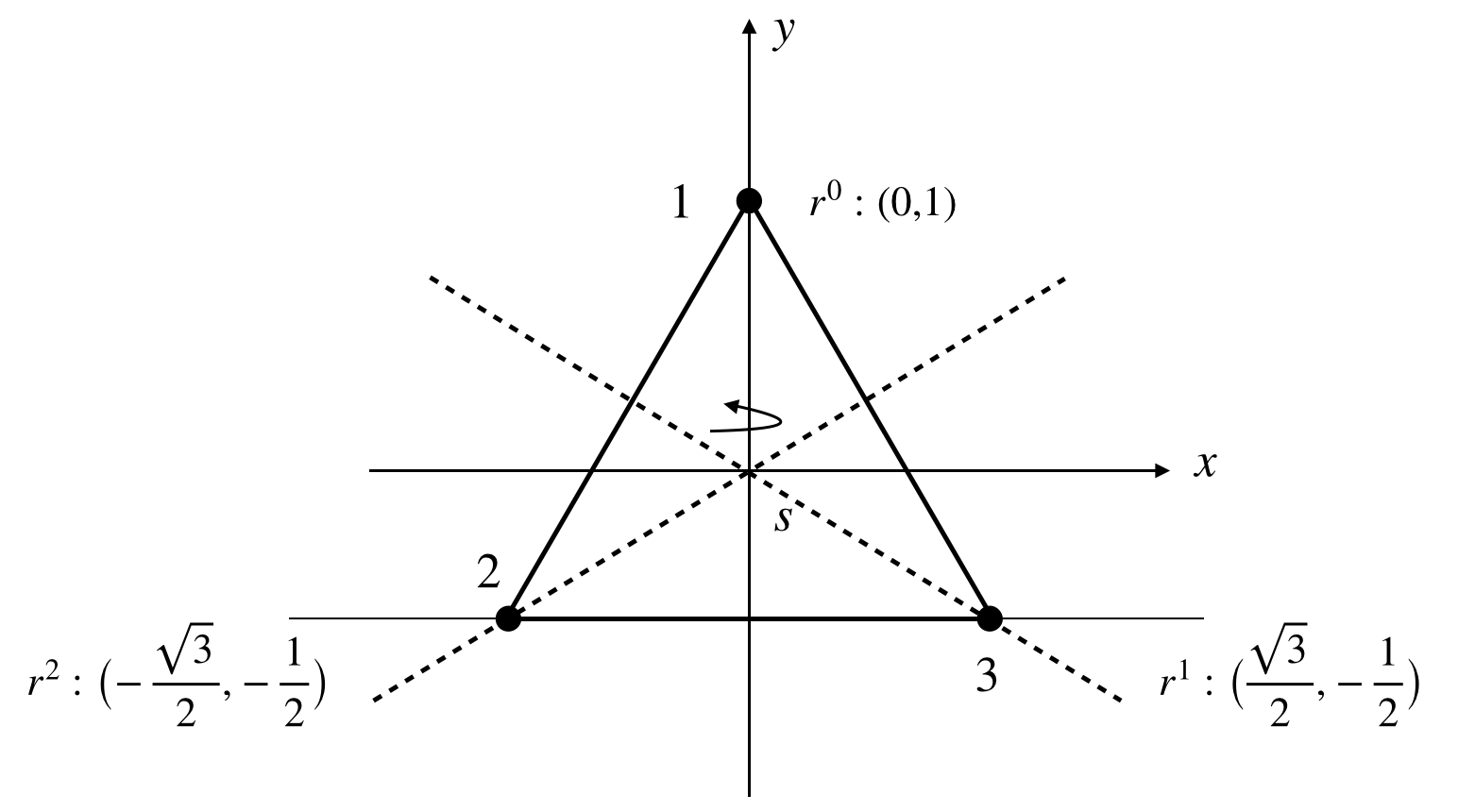}
\caption{The dihedral group $D_3$ is the symmetry group of an equilateral triangle, that is, it is the set of all transformations such as reflection, rotation, and combinations of these, that leave the shape and position of this triangle fixed.}
\label{fig:d3}
\end{figure}

\begin{table}
\centering
\caption{The Cayley table of $D_3$.}
\begin{ruledtabular}
\begin{tabular}
{ c   c  c  c  c  c  c }

& 1 & $r$ & $r^2$ & $s$ & $rs$ & $r^2s$ \\
\hline
1 & 1 & $r$ & $r^2$ & $s$ & $rs$ & $r^2s$ \\

$r$ & $r$ & $r^2$ & 1 & $rs$ & $r^2s$ & $s$ \\ 

$r^2$ & $r^2$ & 1 & $r$ & $r^2s$ & $s$ & $rs$ \\ 

$s$ & $s$ & $r^2s$ & $rs$ & $1$ & $r^2$ & $r$ \\ 

$rs$ & $rs$ & $s$ & $r^2s$ & $r$ & $1$ & $r^2$  \\ 

$r^2s$ & $r^2s$ & $rs$ & $s$ & $r^2$ & $r$ & 1 \\

\end{tabular}
\end{ruledtabular}
\label{table:cayley}
\end{table}

\paragraph{Left and right regular representations of $D_3$.} The regular representations of $D_3$ are obtained by associating a basis vector to each element of the group $\{1, r, r^2, s, rs, r^2s\}$.
\begin{equation*}
    \vec{e}_1 = \left( \begin{array}{ccc}
    1\\
    0\\
    0\\
    0\\
    0\\
    0
    \end{array} \right) \qquad
    \vec{e}_r = \left( \begin{array}{ccc}
    0\\
    1\\
    0\\
    0\\
    0\\
    0
    \end{array} \right) \qquad
    \vec{e}_{r^2} = \left( \begin{array}{ccc}
    0\\
    0\\
    1\\
    0\\
    0\\
    0
    \end{array} \right) \qquad
    \vec{e}_{s} = \left( \begin{array}{ccc}
    0\\
    0\\
    0\\
    1\\
    0\\
    0
    \end{array} \right) \qquad
    \vec{e}_{rs} = \left( \begin{array}{ccc}
    0\\
    0\\
    0\\
    0\\
    1\\
    0
    \end{array} \right) \qquad
    \vec{e}_{r^2 s} = \left( \begin{array}{ccc}
    0\\
    0\\
    0\\
    0\\
    0\\
    1
    \end{array} \right)
\end{equation*}

For all $u \in D_3$, the left regular representations $L_u$ are:
\begin{equation*}
    L_1 = \left( \begin{array}{cccccc}
    1 & 0 & 0 & 0 & 0 & 0\\
    0 & 1 & 0 & 0 & 0 & 0\\
    0 & 0 & 1 & 0 & 0 & 0\\
    0 & 0 & 0 & 1 & 0 & 0\\
    0 & 0 & 0 & 0 & 1 & 0\\
    0 & 0 & 0 & 0 & 0 & 1
    \end{array} \right) \qquad
    L_r = \left( \begin{array}{cccccc}
    0 & 0 & 1 & 0 & 0 & 0\\
    1 & 0 & 0 & 0 & 0 & 0\\
    0 & 1 & 0 & 0 & 0 & 0\\
    0 & 0 & 0 & 0 & 0 & 1\\
    0 & 0 & 0 & 1 & 0 & 0\\
    0 & 0 & 0 & 0 & 1 & 0
    \end{array} \right) \qquad
    L_{r^2} = \left( \begin{array}{cccccc}
    0 & 1 & 0 & 0 & 0 & 0\\
    0 & 0 & 1 & 0 & 0 & 0\\
    1 & 0 & 0 & 0 & 0 & 0\\
    0 & 0 & 0 & 0 & 1 & 0\\
    0 & 0 & 0 & 0 & 0 & 1\\
    0 & 0 & 0 & 1 & 0 & 0
    \end{array} \right)
\end{equation*}

\begin{equation*}
    L_s = \left( \begin{array}{cccccc}
    0 & 0 & 0 & 1 & 0 & 0\\
    0 & 0 & 0 & 0 & 0 & 1\\
    0 & 0 & 0 & 0 & 1 & 0\\
    1 & 0 & 0 & 0 & 0 & 0\\
    0 & 0 & 1 & 0 & 0 & 0\\
    0 & 1 & 0 & 0 & 0 & 0
    \end{array} \right) \qquad
    L_{rs} = \left( \begin{array}{cccccc}
    0 & 0 & 0 & 0 & 1 & 0\\
    0 & 0 & 0 & 1 & 0 & 0\\
    0 & 0 & 0 & 0 & 0 & 1\\
    0 & 1 & 0 & 0 & 0 & 0\\
    1 & 0 & 0 & 0 & 0 & 0\\
    0 & 0 & 1 & 0 & 0 & 0
    \end{array} \right) \qquad
    L_{r^2s} = \left( \begin{array}{cccccc}
    0 & 0 & 0 & 0 & 0 & 1\\
    0 & 0 & 0 & 0 & 1 & 0\\
    0 & 0 & 0 & 1 & 0 & 0\\
    0 & 0 & 1 & 0 & 0 & 0\\
    0 & 1 & 0 & 0 & 0 & 0\\
    1 & 0 & 0 & 0 & 0 & 0
    \end{array} \right).
\end{equation*}

Similarly, right regular representations $R_u$ of $D_3$ are:
\begin{equation*}
    R_1 = \left( \begin{array}{cccccc}
    1 & 0 & 0 & 0 & 0 & 0\\
    0 & 1 & 0 & 0 & 0 & 0\\
    0 & 0 & 1 & 0 & 0 & 0\\
    0 & 0 & 0 & 1 & 0 & 0\\
    0 & 0 & 0 & 0 & 1 & 0\\
    0 & 0 & 0 & 0 & 0 & 1
    \end{array} \right) \qquad
    R_r = \left( \begin{array}{cccccc}
    0 & 0 & 1 & 0 & 0 & 0\\
    1 & 0 & 0 & 0 & 0 & 0\\
    0 & 1 & 0 & 0 & 0 & 0\\
    0 & 0 & 0 & 0 & 1 & 0\\
    0 & 0 & 0 & 0 & 0 & 1\\
    0 & 0 & 0 & 1 & 0 & 0
    \end{array} \right) \qquad
    R_{r^2} = \left( \begin{array}{cccccc}
    0 & 1 & 0 & 0 & 0 & 0\\
    0 & 0 & 1 & 0 & 0 & 0\\
    1 & 0 & 0 & 0 & 0 & 0\\
    0 & 0 & 0 & 0 & 0 & 1\\
    0 & 0 & 0 & 1 & 0 & 0\\
    0 & 0 & 0 & 0 & 1 & 0
    \end{array} \right)
\end{equation*}

\begin{equation*}
    R_s = \left( \begin{array}{cccccc}
    0 & 0 & 0 & 1 & 0 & 0\\
    0 & 0 & 0 & 0 & 1 & 0\\
    0 & 0 & 0 & 0 & 0 & 1\\
    1 & 0 & 0 & 0 & 0 & 0\\
    0 & 1 & 0 & 0 & 0 & 0\\
    0 & 0 & 1 & 0 & 0 & 0
    \end{array} \right) \qquad
    R_{rs} = \left( \begin{array}{cccccc}
    0 & 0 & 0 & 0 & 1 & 0\\
    0 & 0 & 0 & 0 & 0 & 1\\
    0 & 0 & 0 & 1 & 0 & 0\\
    0 & 0 & 1 & 0 & 0 & 0\\
    1 & 0 & 0 & 0 & 0 & 0\\
    0 & 1 & 0 & 0 & 0 & 0
    \end{array} \right) \qquad
    R_{r^2s} = \left( \begin{array}{cccccc}
    0 & 0 & 0 & 0 & 0 & 1\\
    0 & 0 & 0 & 1 & 0 & 0\\
    0 & 0 & 0 & 0 & 1 & 0\\
    0 & 1 & 0 & 0 & 0 & 0\\
    0 & 0 & 1 & 0 & 0 & 0\\
    1 & 0 & 0 & 0 & 0 & 0
    \end{array} \right)
\end{equation*}

\paragraph{Irreducible representations of $D_3$} As $D_3$ is a non-abelian group, at least one of its irreducible representations is a matrix. \autoref{table:d3irreps} shows the irreducible representations of $D_3$, obtained from \autoref{eq:1d-irrep} and \autoref{eq:2d-irrep}, noting that $h \in \{1,2, \ldots, \lceil\frac{n}{2}\rceil -1 \} = \{ 1 \} $.

\begin{table}
\centering
\squeezetable
\caption{Irreducible representations of $D_3$}
\begin{ruledtabular}
\begin{tabular}{ c  c  c  c }
& $\rho_1$ & $\rho_2$ & $\rho_3$ \\ 
& $\sigma_{tt}((x,a))$ & $\sigma_{ts}((x,a))$ & $\sigma_{1}((x,a))$ \\
\hline
(0,0) & 1 & 1 & $\begin{pmatrix} 1 & 0 \\ 0 & 1 \end{pmatrix}$ \\ 

(1,0) & 1 & 1 & $\begin{pmatrix} \omega^1 & 0 \\ 0 & \omega^{-1} \end{pmatrix}$ \\ 

(2,0) & 1 & 1 & $\begin{pmatrix} \omega^2 & 0 \\ 0 & \omega^{-2} \end{pmatrix}$ \\ 

(0,1) & 1 & -1 & $\begin{pmatrix} 0 & 1 \\ 1 & 0 \end{pmatrix}$ \\ 

(1,1) & 1 & -1 & $\begin{pmatrix} 0 & \omega^{1} \\ \omega^{-1} & 0 \end{pmatrix}$ \\

(2,1) & 1 & -1 & $\begin{pmatrix} 0 & \omega^{2} \\ \omega^{-2} & 0 \end{pmatrix}$ \\ 

\end{tabular}
\end{ruledtabular}
\label{table:d3irreps}
\end{table}

\subsection{The Group Fourier Transform and the Convolution Theorem over $D_3$}
The group Fourier transform table (i.e. $F_G$) for $D_3$ can be constructed by aligning the elements of the 2-dimensional representation $\rho_3$ elementwise $(\rho_{3_{11}}, \rho_{3_{12}}, \rho_{3_{21}}, \rho_{3_{22}})$, yielding a $6 \times 6$ transformation matrix.

The normalized (unitary) Fourier transformation matrix $F_G$ is defined as \cite{childs2010quantum}
\begin{equation*}
    F_G = \sum_{x \in G} \ket{\hat{x}}\bra{x} = \sum_{x \in G}\sum_{\rho \in \hat{G}}\sqrt{\frac{d_{\rho}}{|G|}}\sum_{j,k =1}^{d_{\rho}}\rho(x)_{j,k}\ket{\rho, j,k}\bra{x}
\end{equation*}
where $\hat{G}$ is the set of irreducible representations and the $\sqrt{\frac{d_{\rho}}{|G|}}$ factor enforces $F_G$ as unitary. For $D_3$, we have:

\begin{equation}
F_G =
\begin{pmatrix}
1/\sqrt{6} & 1/\sqrt{6} & 1/\sqrt{6} & 1/\sqrt{6} & 1/\sqrt{6} & 1/\sqrt{6}\\
1/\sqrt{6} & 1/\sqrt{6} & 1/\sqrt{6} & -1/\sqrt{6} & -1/\sqrt{6} & -1/\sqrt{6}\\
1/\sqrt{3} & \omega^{1}/\sqrt{3} & \omega^{2}/\sqrt{3} & 0 & 0 & 0 \\
0 & 0 & 0 & 1/\sqrt{3} & \omega^{1}/\sqrt{3} & \omega^{2}/\sqrt{3} \\
0 & 0 & 0 & 1/\sqrt{3} & \omega^{-1}/\sqrt{3} & \omega^{-2}/\sqrt{3} \\
1/\sqrt{3} & \omega^{-1}/\sqrt{3} & \omega^{-2}/\sqrt{3} & 0 & 0 & 0
\label{fouriermatrix}
\end{pmatrix}.
\end{equation}

Let $m$ and $f$ be functions that map group elements of $G$ to complex numbers, if we associate each element $u \in G$ to a basis vector $e_u$ in some vector space $V$, we can represent $m$ and $f$ as vectors,
\begin{equation*}
    \vec{m} = \sum_{u \in G} m(u)e_u, \qquad 
    \vec{f} = \sum_{u \in G} f(u)e_u.
\end{equation*} 

As an example, we take
\begin{equation}
\vec{m} = \vec{f} =
\begin{pmatrix}
    1\\
    \omega^1\\
    \omega^2\\
    0\\
    0\\
    0
\end{pmatrix},
\end{equation}
where we have chosen $e_u$ to be the standard basis of $\mathbb{C}^{|G|}$.

Calculating the  Fourier transform through matrix multiplication on $\vec{m}$ and $\vec{f}$ we obtain
\begin{equation*}
\widehat{m} = \widehat{f} =  F_G \vec{m} = F_G \vec{f} = 
\left( \begin{array}{ccc}
0\\
0\\
0\\
0\\
0\\
\sqrt{3} \end{array} \right)
\end{equation*}
Note that $1 + \omega^1 + \omega^2 = 0$ since the cube root of unity (i.e. $\omega^3 = 1$) can be factorized as 
\begin{equation*}
    \omega^3 - 1 = (\omega-1)(\omega^2+\omega+1) = 0.
\end{equation*}

In the following we compare the computation of the Fourier transform of $\vec{m} \circledast \vec{f}$ by two different methods: (1) through direct calculation of the convolution and applying the Fourier transform,
and (2) by computing the individual Fourier transforms of $\vec{m}$ and $\vec{f}$ via the convolution theorem.\\
First, we complete case (1). Recall the definition of a convolution over a group $G$
\begin{equation*}
    (m \circledast f)(u) = \sum_{v \in G} m(uv^{-1})f(v)
\end{equation*}
Calculating the convolution expansion over the indexed elements of $D_3$, where ${1,2,3,4,5,6}$ correspond to  (refer to \autoref{table:cayley} for group element multiplication $uv^{-1}$)
\begin{equation*}
\begin{split}
    (m\circledast f)(1) =  m(1) \cdot f(1) + m(3)\cdot f(2) + m(2)\cdot f(3)+m(4)\cdot f(4) + m(5)\cdot f(5) + m(6)\cdot f(6)\\
    (m\circledast f) (2) = m(2) \cdot f(1) + m(1)\cdot f(2) + m(3)\cdot f(3)+m(5)\cdot f(4) + m(6)\cdot f(5) + m(4)\cdot f(6)\\
    (m\circledast f) (3) = m(3) \cdot f(1) + m(2)\cdot f(2) + m(1)\cdot f(3)+m(6)\cdot f(4) + m(4)\cdot f(5) + m(5)\cdot f(6)\\
    (m\circledast f) (4) = m(4)\cdot f(1) + m(5)\cdot f(2) + m(6)\cdot f(3)+m(1)\cdot f(4) + m(3)\cdot f(5) + m(2)\cdot f(6)\\
    (m\circledast f) (5) = m(5)\cdot f(1) + m(6)\cdot f(2) + m(4)\cdot f(3)+m(2)\cdot f(4) + m(1)\cdot f(5) + m(3)\cdot f(6) \\
    (m\circledast f) (6) = m(6)\cdot f(1) + m(4)\cdot f(2) + m(5)\cdot f(3)+m(3)\cdot f(4) + m(2)\cdot f(5) + m(1)\cdot f(6).
\end{split}    
\end{equation*}

In particular, for $m$ and $f$
\begin{equation*}
\begin{split}
    (m \circledast f)(1) =&  1 \cdot 1 + \omega^2\cdot\omega^1 + \omega^1\cdot\omega^2 + 0 \cdot 0 + 0 \cdot 0 + 0 \cdot 0 = 3 \\
    (m \circledast f)(2) =&  \omega^1 \cdot 1 + 1\cdot\omega^1 + \omega^2\cdot\omega^2 + 0 \cdot 0 + 0 \cdot 0 + 0 \cdot 0 = 3\omega^1\\
    (m \circledast f)(3) =& \omega^2 \cdot 1 + \omega^1\cdot\omega^1 + 1\cdot\omega^2 + 0 \cdot 0 + 0 \cdot 0 + 0 \cdot 0 = 3\omega^2\\
    (m \circledast f)(4) =& (m \circledast f)(5) = (m \circledast f)(6) = 0.
\end{split}
\end{equation*}
Written in vector form,
\begin{equation*}
    \vec{m} \circledast \vec{f} = \left( \begin{array}{ccc}
    3\\
    3\omega^1\\
    3\omega^2\\
    0\\
    0\\
    0
    \end{array} \right),
\end{equation*}
calculating the Fourier transform of $\vec{m} \circledast \vec{f}$ through matrix multiplication with $F_G$,

\begin{equation*}
\widehat{m \circledast f} = F_G(\vec{m} \circledast \vec{f}) = \left( \begin{array}{ccc}
    0\\
    0\\
    0\\
    0\\
    0\\
    3\sqrt{3}
    \end{array} \right).
\label{eq:ex-result1}
\end{equation*}

Now, we overview case (2). Recall the convolution theorem (\autoref{eq:convolution_theorems})
\begin{equation*}
    (\widehat{m \circledast f})(\rho_i) = \widehat{m}(\rho_i)\widehat{f}(\rho_i).
\end{equation*}

First we compute the group Fourier transform (\autoref{eq:groupfour}) of $m$ and $f$ over the irreducible representations,
\begin{equation*}
\widehat{m}(\rho_1) = \widehat{f}(\rho_1) = 0 \qquad
\widehat{m}(\rho_2) = \widehat{f}(\rho_2) = 0 \qquad
\widehat{m}(\rho_3) = \widehat{f}(\rho_3) = \left( \begin{array}{ccc}
0 & 0\\
0 & 3 \end{array} \right).
\end{equation*}

Applying the convolution theorem, we have
\begin{equation*}
\widehat{m \circledast f}(\rho_1) = \widehat{m}(\rho_1)\widehat{f}(\rho_1) = 0 \qquad
\widehat{m \circledast f}(\rho_2) = \widehat{m}(\rho_2)\widehat{f}(\rho_2) = 0 \qquad
\widehat{m \circledast f}(\rho_3) = \widehat{m}(\rho_3)\widehat{f}(\rho_3) = \left( \begin{array}{ccc}
0 & 0\\
0 & 9 \end{array} \right) 
\qquad
\end{equation*}

Aligning the elements of $\widehat{m \circledast f}(\rho_i)$ in order for all the irreps $\rho_i$ on a 6-dim vector on the same standard basis we obtain 

\begin{equation}
    \widehat{m \circledast f}= 
    \left(\begin{array}{ccc}
    0\\
    0\\
    0\\
    0\\
    0\\
    9
    \end{array} \right)
\label{eq:ex-result2}
\end{equation}



Note that \autoref{eq:ex-result1} and \autoref{eq:ex-result2} differ by a factor of normalization $\sqrt{d_{\rho}/|G|}$ since in method (1) $F_G$ is already forced to be unitary, while method (2) is based on purely classical calculations.

Alternatively, we can also perform the group Fourier transform in the block encoding form of \autoref{eq:nonabelian_general},

\begin{equation}
    \widehat{m \circledast f} = 
    \Big[ \bigoplus_{v \in G} \widehat{m}(\rho_i)\otimes I_{d_{\rho}} \Big] \widehat{f}
    =
    \begin{pmatrix}
    \begin{bmatrix}
    0
    \end{bmatrix} \oplus
    \begin{bmatrix}
    0
    \end{bmatrix} \oplus
    \begin{bmatrix}
    0 & 0 \\
    0 & 3
    \end{bmatrix} \otimes
    I_2
    \end{pmatrix}
    \begin{pmatrix}
    0\\
    0\\
    0\\
    0\\
    0\\
    \sqrt{3}
    \end{pmatrix} =
    \begin{pmatrix}
    0\\
    0\\
    0\\
    0\\
    0\\
    3\sqrt{3}
    \end{pmatrix}
\end{equation}
which already includes the normalization factor on $F_G$.


\section{Equivariance}
\label{app:equivariance}

Here, we show explicitly that convolution and cross-correlation are equivariant actions. Let $G$ be a group and $\mathcal{X}_1, \mathcal{X}_2$ be two sets with corresponding $G$-actions
\begin{equation*}
    T_g : \mathcal{X}_1 \rightarrow \mathcal{X}_1 \;\;\;\;\;\; T'_g: \mathcal{X}_2 \rightarrow \mathcal{X}_2.
\end{equation*}

Let $V_1$ and $V_2$ be vector spaces with basis elements labeled by elements of $\mathcal{X}_1$ and $\mathcal{X}_2$, respectively, and let $L_{V_1}, (L_{V_2})$ be the set of functions mapping $\mathcal{X}_1(\mathcal{X}_2)$ to $V_1 (V_2)$. First, we will look at the case of convolution.

\subsection{Convolution}
Let $\phi_m : L_{V_1} \rightarrow L_{V_2} $ be the map performing convolution with a fixed filter $\vec{m}$ on an input  $\vec{f}$,
\begin{equation*}
    \phi_m(\vec{f}) = \vec{m} \circledast \vec{f}.
\end{equation*}
Let $T_g, T'_g$ denote the right actions of the group,
\begin{equation}
\begin{split}
    T_g, T'_g &: u \rightarrow ug,
\end{split}
\end{equation}
and let $\mathbb{T}_g$ and $\mathbb{T}'_g$ be the induced actions of group elements onto $V_1$ and $V_2$ respectively. From definition 1.1, the map $\phi_m: L
_{V_1} \rightarrow L
_{V_2}$ is \emph{equivariant} to the action of $T_g$ since 
\begin{equation*}
    \phi_m(\mathbb{T}_g \vec{f}) = \mathbb{T}_g'( \phi_m(\vec{f})).
\end{equation*}



\begin{proof}
Recall the convolution definition from \autoref{eq:convolution_theorems}
\begin{equation*}
    \Big[\phi_m (\vec{f}) \Big]_u = \Big[\vec{m} \circledast \vec{f} \Big]_u = \sum_{v \in G} m(uv^{-1})f(v)
\end{equation*}

Let $\phi_m$ act on $\mathbb{T}_g \vec{f}$,
\begin{equation*}
\begin{split}
    \Big[\phi_m (\mathbb{T}_g \vec{f}) \Big]_u = \Big[ \vec{m} \circledast \mathbb{T}_g\vec{f} \Big]_u &= \sum_{v\in G} m(uv^{-1}) \Big[\mathbb{T}_g \vec{f} \Big]_v\\
    &= \sum_{v\in G} m(uv^{-1})f(vg)
\end{split}
\end{equation*}


Now redefine the sum above over $v' = vg$, and now we have $v^{-1} = g(v')^{-1}$, 
\begin{equation}
\begin{split}
    \sum_{v \in G} m(uv^{-1})f(vg) &= \sum_{v' \in G} m(ugv'^{-1})f(v')\\
    &= \Big[\vec{m} \circledast \vec{f}\Big]_{ug}\\
    & = \Big[\phi_m (\vec{f})\Big]_{ug}\\
    & = \Big[\mathbb{T}'_g(\phi_m(\vec{f}))\Big]_u \;\;\;\;\; \forall \; u \in G.
\end{split}
\end{equation}

Concluding that $\phi_m(\mathbb{T}_g \vec{f}) = \mathbb{T}'_g(\phi_m(\vec{f}))$; hence convolution is equivariant to the right actions of the group.
\end{proof}

\subsection{Cross-correlation}
Let $\phi_m: L_{V_1} \rightarrow L_{V_2}$ be the map performing cross-correlation with a fixed filter $\vec{m}$ on an input $\vec{f}$,

\begin{equation*}
    \phi_m (\vec{f}) = \vec{m} \star \vec{f}.
\end{equation*}
Let $T_g$ and $T'_g$ denote the right actions of the group,

\begin{equation*}
    T_g, T'_g : u \rightarrow ug
\end{equation*}
Let $\mathbb{T}_g$ and $\mathbb{T}'_g$ be the induced action of group elements onto $V_1$ and $V_2$. From definition 1.1, the map $\phi_m: L_{V_1} \rightarrow L_{V_2}$ is equivariant to $T_g$ since 
\begin{equation*}
    \phi_m(\mathbb{T}_g\vec{f}) = \mathbb{T}_g' (\phi_m(\vec{f})).
\end{equation*}

\begin{proof}
    Recall the cross correlation definition from \autoref{eq:convolution_theorems}
\begin{equation*}
    \Big[\phi_m(\vec{f})\Big]_u = \Big[\vec{m} \star \vec{f}\Big]_u = \sum_{v\in G} m(vu^{-1})f(v).
\end{equation*}

Let $\phi_m$ act on $\mathbb{T}_g\vec{f}$,
\begin{equation*}
\begin{split}
    \Big[\phi_m (\mathbb{T}_g \vec{f}) \Big]_u = \Big[\vec{m} \star \mathbb{T}_g \vec{f} \Big]_u &= \sum_{v\in G} m(uv)\Big[\mathbb{T}_g\vec{f}\Big]_v \\
    &= \sum_{v\in G} m(vu^{-1})f(vg)
\end{split}
\end{equation*}
Now redefine the sum above over $v' = vg$, and now we have $v = v'g^{-1}$,
\begin{equation*}
\begin{split}
    \sum_{v\in G} m(vu^{-1})f(vg) &= \sum_{v'\in G} m(v'g^{-1}u^{-1})f(v')\\
    &= \Big[\vec{m} \star \vec{f} \Big]_{ug}\\
    &= \Big[ \mathbb{T}'_g (\phi_m(\vec{f})) \Big]_u \;\;\;\;\; \forall u \in G.
\end{split}
\end{equation*}
Concluding that $\phi_m(\mathbb{T}_g\vec{f}) = \mathbb{T}'_g(\phi_m(\vec{f}))$; hence cross-correlation is equivariant to the right actions of the group. 
\end{proof}

\section{Details of algorithm for solving integral equation} \label{app:integraleqn}

In this section we detail an application of our quantum group convolution algorithms for solving integral equations where an unknown function appears under an integral sign. Integral equations often arise in mathematical physics models which require summations (integrations) over space or time or can be formed by appropriately transforming differential operators \cite{source_integral_eq}. Integral equations have wide applications in scientific and engineering problems including diffraction problems, scattering in quantum mechanics, and conformal mappings, among others. One common integral equation is a second kind Fredholm equation shown in general form below:
\begin{equation}
    g(\vec{x}) = y(\vec{x}) f(\vec{x}) + \lambda \int_{\vec{a}}^{\vec{b}} K(\vec{x}, \vec{x}') f(\vec{x}') d\vec{x}',
\end{equation}
where $f(\vec{x})$ is the unknown function and $K(\vec{x},\vec{x}')$ is the kernel of the integral equation. In physical settings, kernels are often potential functions (e.g. gravitational, electric potentials) which depend only on the difference between their arguments: $K (\vec{x},\vec{x}')) =K(\vec{x}-\vec{x}'))$, commonly termed displacement kernels. Kernels commonly used in machine learning such as the squared exponential or rational quadratic kernels also take this form \cite{rasmussen2003gaussian}.

For most integral equations, closed form solutions do not exist. Thus, one must numerically solve integral equations by converting them into a linear system of equations. One option for discretizing an integral equation is the Nystr\"om method \cite{atkinson2009numerical} which approximates the integral via a weighted summation over discretized points (\textit{e.g.,} quadrature or Riemann sum) and solves for the values of the unknown function at discretized points. If certain group symmetries are present in the problem setting, one can use the Nystr\"om method to convert the numerical integration into a finite group cross-correlation (recall definitions in  \Cref{lem:conv_to_linear_comb}) and solve the integral equation using our quantum algorithms in time logarithmic in the number of discretized points given inputs stored as quantum data.

We consider an integral equation defined over the boundary of a $2$-dimensional periodic (wrap-around) lattice:
\begin{equation}
    g(t_1,t_2) = f(t_1, t_2) + \lambda \int_{0}^1 \int_{0}^1 K(\| \vec{t} - \vec{t'} \|) f(t_1', t_2') dt_1' dt_2'.
    \label{app:equation}
\end{equation}
Here, we would like to solve for $f(t_1,t_2)$. Due to the periodic property of the lattice, if we discretize $t_1$ and  $t_2$ evenly, the resultant discritized grid is invariant to certain two dimensional translations of the grid (i.e., translations from the direct sum of two cyclic groups). In this setting, we will show that numerical discretization of integral equations with a displacement kernel $K(\|\vec{t}-\vec{t'}\|)$ take the form of group convolution/cross-correlation. Using the Nystr\"om method, we approximate the integral as a weighted sum over $n^2$ discretized points evenly spaced on $t_1$ and $t_2$:
\begin{equation}
    \int_0^1 \int_0^1 K(\| \vec{t} - \vec{t'} \|) f(\vec{t'}) d\vec{t'}  \approx \frac{1}{n^2} \sum_{i_1 = 1}^{n} \sum_{i_2=1}^{n} K(\| \vec{t} - \vec{t'}_{i_1, i_2} \| ) f(\vec{t'}_{i_1, i_2}).
\end{equation}
Here we have utilized the trapezoidal rule, noting that the integrand evaluates to the same values at the boundaries of the lattice. \autoref{app:equation} can then be written as a system of equations
\begin{equation}
    \boldsymbol{g} \approx (I+\lambda \boldsymbol{K}) \boldsymbol{f} = \boldsymbol{\widetilde{K}} \boldsymbol{f},
    \label{eq:app.discretized}
\end{equation}
where bold text indicates discretized values of the functions vectorized in lexicographic order and $\boldsymbol{K}$ is an $n^2 \times n^2$ matrix whose entries $\boldsymbol{K}_{ij} = (1/n^2) K(\| \vec{t}_i - \vec{t'_j} \|)$ (note that $i, j \in [n^2] $ here index vectors $\vec{t}$ or $\vec{t'}$ at grid points). The form of the matrix $\boldsymbol{\widetilde{K}}$ (and $\boldsymbol{K}$) corresponds to a group cross-correlation over the direct product of the cyclic groups $(\mathbb{Z}/n\mathbb{Z}) \times (\mathbb{Z}/n\mathbb{Z})$ where the filter $\boldsymbol{m}$ is equal to the values of the first row (or equivalently the first column) of $\boldsymbol{\widetilde{K}}$. In other words, moving from one row of $\boldsymbol{\widetilde{K}}$ to another is equivalent to applying a permutation operation corresponding to one of the group elements. In our case, these permutations are the (left) regular representations of $(\mathbb{Z}/n\mathbb{Z}) \times (\mathbb{Z}/n\mathbb{Z})$:
\begin{equation}
    L_{i} = P_{i_1} \otimes P_{i_2} \text{,  for } i \in [n^2],
\end{equation}
where $i_1 = \lfloor i/n \rfloor$, $i_2 = i\mod n$, and $P_{i}$ are the $n \times n$ cyclic permutation matrices corresponding to the cyclic group $\mathbb{Z}/n\mathbb{Z}$: $[P_i]_{jk}=\delta_{j, k+i}$.

As stated above, we take the filter $\boldsymbol{m}$ to be the first column of $\boldsymbol{\widetilde{K}}$: 
\begin{equation}
    m_i = \delta_{i0}+\lambda \boldsymbol{K}_{i0}, \text{ for } i \in [n^2],
    \label{eq:app.filter}
\end{equation}
 and rewrite $\boldsymbol{\widetilde{K}} = \sum_{i \in [n^2]}m_i L_i^{-1}$ (\Cref{lem:conv_to_linear_comb}). Thus, \autoref{eq:app.discretized} is equivalent to a group cross-correlation 
 \begin{equation}
      \boldsymbol{g} \approx \boldsymbol{m} \star \boldsymbol{f}.
      \label{eq:app.cross-correlation}
 \end{equation}
 
We now show that the above procedure is numerically robust by analyzing a numerical experiment where we can conveniently calculate errors by comparing to a closed form solution (of course, the existence of a closed form solution is not required for our algorithm). We emphasize here that we are directly inverting the cross-correlation matrix $\boldsymbol{\widetilde{K}}$ in \autoref{eq:app.discretized} and not simulating the experiment as it would be performed on a quantum computer. Later, we discuss runtimes and errors should on perform this task one a quantum computer. We take the exponential kernel
\begin{equation}
    K(\| \vec{t} - \vec{t'} \|) = e^{-D( \vec{t}, \vec{t'}) },
    \label{eq:app.kernel}
\end{equation}
where $D(\vec{t}, \vec{t'}) = \sum_{i}^{d} \min \{| t_i - t_i' | , L_i - | t_i - t_i' | \}$ is the Manhattan distance \cite{taxicab} defined over the periodic lattice of size $L_i$ in each dimension (in our case $d=2$ and $L_i = 1$),
\begin{equation}
 \lambda = 1, \qquad \qquad g(t_1,t_2) = (t_1 - t_1^3)(t_2 - t_2^3) + h(t_1)h(t_2),
 \label{eq:app.example}
\end{equation}
where 
\begin{equation}
    h(u) = \begin{cases} 
      -3 e^{-u} -2u(5+u^2) + \frac{(1+2u)(21+4u(1+u)) }{4\sqrt{e}} \qquad \text{if } 0 \leq u < 1/2\\
      9 e^{-1+u} -2u(5+u^2) + \frac{(-1+2u)(21+4u(-1+u)) }{4\sqrt{e}} \qquad \text{if } 1/2 \leq u < 1
   \end{cases}.
\label{eq:app.examples}
\end{equation}
The solution takes the following simple form
\begin{equation}
    f(t_1,t_2)=(t_1 - t_1^3)(t_2 - t_2^3).
\end{equation}

\autoref{fig:solution} displays the the underlying solution alongside the discretized solution $\boldsymbol{f}$ for various numbers of discretized points $n^2$, where $n \in \{ 4,16,64\}$. We observe that the numerical solution using inverse group cross-correlation converges to the true solution. Furthermore, as numerically shown in \autoref{fig:error}, the average absolute error, defined as
\begin{equation}
    \frac{\|\boldsymbol{f}-f(\boldsymbol{t})\|_1}{n^2},
    \label{eq:app.error}
\end{equation}
decreases at a rate of $ 1/n^2$ with the discretization resolution $n$. Here, we use the notation $f(\boldsymbol{t})$ to represent the vector of values of the underlying solution at the discretized points. This is consistent with errors in integral approximations with the trapezoidal rule where $n$ discretized points in each dimension produces an error of $O(1/n^2)$ \cite{atkinson2009numerical} (for functions with non-smooth derivatives at the boundary). In addition, the condition number of $\boldsymbol{\widetilde{K}}$ is bounded by $O(1+1/n^2)$, as we will show later. Therefore, by the fundamental theorem of numerical analysis \cite{Arnold2015}, the error in convergence to the true solution decays at a rate of $1/n^2$. For details of these error analysis techniques, we refer the reader to \cite{atkinson2009numerical}.

\begin{figure}[h]
\centering
\includegraphics{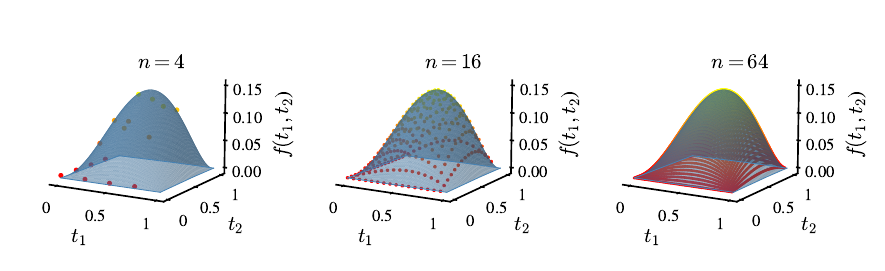}
\caption[Caption for LOF]{Numerical solution (yellow-red dots) of the integral equation defined in \autoref{app:equation} and \autoref{eq:app.examples} using our Nystr\"om cross-correlation approach shown side-by-side with the underlying solution $f(t_1,t_2) = (t_1 -t_1^3)(t_2-t_2^3)$ (blue surface) for increasing number of discretized points $n$ along each dimension\footnote{See supplementary code at: \url{https://github.com/nguyenquantum/group-convolution}.}.}
\label{fig:solution}
\end{figure}

\begin{figure}[h]
\centering
\includegraphics[width=8.6cm]{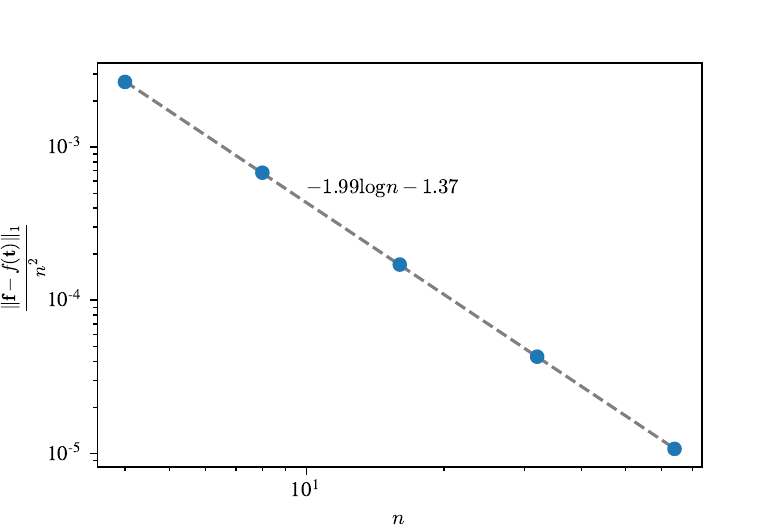}
\caption[Caption for LOF]{Log-log scale: average absolute error (\autoref{eq:app.error}) of our Nystr\"om cross-correlation approach for solving the integral equation decreases at rate $\sim 1/n^2$ (dashed grey), where $n$ is the number of discretized points along each dimension.}
\label{fig:error}
\end{figure}

The cross correlation inversion step above can be implemented in quantum settings using our quantum group convolution/cross-correlation algorithms as follows. One first computes the group Fourier transform $\boldsymbol{\Hat{m}}$ of the filter (\autoref{eq:app.filter}) and prepares a circuit that implements an oracle $O_{\boldsymbol{\Hat{m}}}$ which returns entries of $\boldsymbol{\Hat{m}}$ (\autoref{eq:fourier_entry_oracle} repeated):
\begin{equation}
    O_{\boldsymbol{\Hat{m}}}: \ket{j}\ket{0} \rightarrow \ket{j} \ket{\Hat{m}_j}.
\end{equation}
When $n$ is sufficiently large, we have $\hat{m}_j \approx \int_{0}^{1}\int_{0}^{1}K(\|\vec{t}\|) \chi_{j}(\vec{t}) d\vec{t}$, where $\chi_{j}(\vec{t}) = e^{2\pi i(j_1t_1+j_2t_2)/n}$ (with $j_1=\lfloor j/n \rfloor, j_2 = j \mod n$) is a character of the group $(\mathbb{Z}/n\mathbb{Z}) ^{\times 2}$. The error of this approximation can be similarly analyzed using standard error bounds from quadrature schemes; we refer the reader to \cite{atkinson2009numerical}. This integral can be done analytically with the exponential kernel $K$ defined in \autoref{eq:app.kernel}, hence there exists an efficient circuit for $O_{\boldsymbol{\Hat{m}}}$. In addition, this oracle only depends on the kernel $K$, the constant $\lambda$, and the number of discretized points $n$ and thus can be reused for different functions $f, g$. Given the input function $g$ stored as a quantum state $|\boldsymbol{g}\rangle$, one can apply our algorithm for inverting group cross-correlation (\Cref{prop:inverse_op} or \Cref{app:fourier-deconvolve}) to efficiently obtain the discretized solution $\boldsymbol{f}$ as a quantum state in time $O(T_B \frac{\kappa^2}{\|\boldsymbol{\widetilde{K}}\|} \operatorname{polylog} \frac{ \kappa}{\|\boldsymbol{\widetilde{K}}\| \epsilon} )$, where we have taken $d_{\max}=1$ as the group $(\mathbb{Z}/n\mathbb{Z})^{\times 2}$ is abelian. Here $T_B$ is the runtime of the block encoding of forward group convolution/cross-correlation in \Cref{lem:abelian_encoding}, which is composed of one group Fourier transform, one inverse group Fourier transform, two calls to the oracle $O_{\boldsymbol{\Hat{m}}}$ above, and $O(\operatorname{polylog} \frac{1}{\epsilon}+\log n)$ additional gates. The group Fourier transform and its inverse can be efficiently implemented by a circuit of size logarithmic in the size of the group \cite{childs2010quantum} (here, $|(\mathbb{Z}/n\mathbb{Z})^{\times 2}| = n^2$). Thus, assuming the two oracle calls also take at worst polylogarithmic time, we have $T_B = O (\operatorname{polylog}{ \frac{1}{\epsilon}, \log n })$.

We now bound the condition number $\kappa$ and the operator norm of the cross-correlation matrix $\boldsymbol{\widetilde{K}}$ in \autoref{eq:app.discretized}. First, the operator norm of $\boldsymbol{K}$ can be bounded as follows:
\begin{equation}
\begin{aligned}
     \|\boldsymbol{K}\| = \|\sum_{i \in [n^2]} \boldsymbol{K}_{i0} L_i^{-1} \| & \leq \sum_{i \in [n^2]} \boldsymbol{K}_{i0} \\
     & \leq \iint K(\|\vec{t} - \vec{0}\|) d\vec{t} + O(1/n^2),
\end{aligned}
\end{equation}
where we have used the triangle inequality, noting that the kernel is real non-negative, and the error estimate of the trapezoidal quadrature scales as $O(1/n^2)$. For example, the integral of the Manhattan-distance-exponential kernel in our example evaluates to $\iint e^{-(D(\vec{t},0))} d\vec{t} = 0.7869$, which is independent of $n$. For the (classical) numerical approximation to be stable, we require that $1- 0.7869|\lambda| = \Omega(1)$ \cite{conditionnumber}, which is indeed the case in our example specified by \autoref{eq:app.example}. The condition number of $\boldsymbol{\widetilde{K}} = I +\lambda \boldsymbol{K}$ can then be bounded as 
\begin{equation}
    \kappa (\boldsymbol{\widetilde{K}}) \leq \frac{1+ |\lambda| \| \boldsymbol{K}\|}{1- |\lambda| \|\boldsymbol{K} \|} \leq \frac{1+ 0.7869 |\lambda|  + O(1/n^2)}{1- 0.7869 |\lambda| - O(1/n^2)} = O(1),
\end{equation}
and the operator norm of $\boldsymbol{\widetilde{K}}$ can be lower bounded as
\begin{equation}
    \|\boldsymbol{\widetilde{K}}\| \geq \sigma_{\min}(\boldsymbol{\widetilde{K}}) \geq 1 - |\lambda| \| K\| = \Omega(1).
\end{equation}

Therefore, given $\boldsymbol{g}$ as quantum data, our quantum inverse cross-correlation algorithm solves the above integral equation in time $O( \log n, \operatorname{polylog} \frac{1}{\epsilon})$ up to normalization. In practice, kernels are often well conditioned for the use of trapezoidal rule, thus the condition number $\kappa$ can be bounded similarly as above. In these well-conditioned settings, our algorithm is applicable to integral equations with any number of group symmetries.

\bibliographystyle{apsrev4-1}
\bibliography{main}

\end{document}